\documentclass[11pt]{article}

\usepackage{graphicx}
\usepackage[square, numbers, comma, sort&compress]{natbib}
\usepackage{graphicx}
\usepackage{tikz}

\usepackage{amsfonts}
\usepackage{amsmath}
\usepackage{amsthm} % for proof environment
\usepackage{amssymb}
\usepackage{epsfig, graphics}
\usepackage{hyperref}
\usepackage{url}
\usepackage{color}
\usepackage{setspace}
\usepackage{float}
\usepackage{subfig}
\usepackage{colortbl}
\usepackage[margin=1cm]{caption}

\usepackage[square, numbers, comma, sort&compress]{natbib}
\usepackage{graphicx}
\usepackage{tikz}

\hypersetup{colorlinks,linkcolor={blue},citecolor={red},urlcolor={blue}}

\newtheorem{theorem}{Theorem}[section]

\newtheorem{corollary}[theorem]{Corollary}

\newtheorem{lemma}[theorem]{Lemma}

\def\EE{\mathbb{E}}
\def\PP{\mathbb{P}}
\def\NN{\mathbb{N}}
\def\RR{\mathbb{R}}

\def\bs{\mathbf{s}}
\def\ind{{\rm 1\hspace{-0.90ex}1}}

\newcommand{\Bin}{\mathsf{Bin}}

\newcommand{\cT}{\mathcal{T}}
\newcommand{\cS}{\mathcal{S}}
\newcommand{\cI}{\mathcal{I}}

\newcommand{\cD}{\mathcal{D}}
\newcommand{\cR}{\mathcal{R}}
\newcommand{\fR}{\mathfrak{R}}
\newcommand{\iI}{\mathcal{I}}

\newcommand{\cG}{\mathcal{G}^{(n)}}

 \def\blue{\color{black}}

\begin{document}

%%%%%%%%%%%%%%%%
\title{Epidemic Spreading and Equilibrium Social Distancing in Heterogeneous Networks}

\author{Hamed Amini\thanks{Georgia State University, Atlanta, GA 30303, USA, email: {\tt hamini@gsu.edu}} \ and Andreea Minca\thanks{Cornell University, Ithaca, NY 14850, USA, email: {\tt acm299@cornell.edu}}}

\date{}
\maketitle
\abstract{We study a multi-type SIR epidemic process within a heterogeneous population that interacts through a network. 
We base social contact on a random graph with given vertex degrees and we give limit theorems on the fraction of infected individuals. For given social distancing individual strategies, we  establish the epidemic reproduction number $\fR_0$, which can be used to identify network vulnerability and inform vaccination policies. 
In the second part of the paper we study the equilibrium of the social distancing game. Individuals choose their social distancing level according to an anticipated global infection rate, which  must equal the actual infection rate following their choices. We give conditions for the existence and uniqueness of an equilibrium.   In the case of random regular graphs, we show that voluntary social distancing will always be socially sub-optimal.  

%Our numerical study using  Covid-19 data serves to quantify the absolute and relative  utility gaps across age cohorts.

\bigskip

}

\bigskip
%\newpage
\section{Introduction}\label{sec:intro}

The coronavirus crisis has seen an unprecedented scale of lockdown measures,  imposed worldwide and in many cases very strictly in order to mitigate the public health threat.
Unarguably, the economic and social impact has been devastating.
The path to reopening the economy remains uncertain,  and outbreaks may re-emerge as soon as measures are relaxed.
While lockdown measures have been shown to have saved a tremendous number of lives, there is little disagreement that they cannot be maintained in the long run. With a virus so contagious,  widespread and prone to variations, the long run will indeed be measured in years rather than months.
The path forward, at least until vaccines are widely available, is more likely to rely on proper guidelines from the governments, targeted quarantines, combined with the transparent information for the public rather than strict and un-targeted lockdowns.

Despite that the public adherence to guidelines -- even if those were optimal--  is the driving force in the post-lockdown world, few epidemic models incorporate individuals' decision-making. One notable exception is \cite{NBERw26984}, who integrate individual decision making in a Susceptible - Infectious - Recovered (SIR) model of contagion. This paper is closest in spirit to ours, and they demonstrate using U.S. micro-data that individuals started to socially distance earlier than the governments mandated to do so. They also investigate the optimal social distancing policy, and show that this policy should be mandated for as long as possible until herd immunity is achieved through vaccines.
Other works focusing specifically on Susceptible - Infectious - Recovered/Dead (SIRD) modeling in the context of this public health  crisis, e.g., \cite{Acemoglu20}, suggest an agenda to make epidemiological models more realistic, and in particular to address multiple sources of heterogeneity. 
First, it is clear that the disease treats people very differently, and, while no age group is spared,  the elderly have significantly worse outcomes. 
At the same time, the contact pattern is nowhere close to the homogeneous mixing of the classical SIR model. In particular, \cite{Acemoglu20} solve the social planner's problem for a multi-type (multi-risk) SIR model and leave for further research the case where  interaction occurs according to a social graph structure.

Our first set of results is concerned with the size of the epidemic when the social interaction occurs on random graphs with given vertex degrees.
%a given social distancing strategy profile across individuals and for a given network underlying the social interactions (that we call interaction graph). An infection matrix is obtained from the adjacency matrix and the type dependent rates with which susceptible individuals seek social contact.
%Based on the spectrum of the infection matrix, we characterize the amplification of the epidemic, namely the ratio between the fraction of infected individuals during the contagion process and the size of the initial seed. In particular, we show that if the largest singular value of the infection matrix is smaller than one, then the expected amplification is of the order $O(\sqrt{n})$ in the size of the network $n$. 
%This represents a testable condition whether a given interaction graph is prone to contagion and can guide governments where to focus an eventual vaccine or identify potential infection hotspots.
%In the same spirit, we  extend our analysis to the case of random graphs underlying the social interaction. 
 We define the network immunity as the probability that a neighbor of a randomly chosen susceptible individual does not get infected during the epidemic.
We impose  mild conditions on the degree distribution of the susceptible population, whereas infectious and recovered individuals' degree distribution can be  more arbitrary. There is a maximum degree condition, which allows for degrees to grow sub-linearly with the network size. Our results are in this case asymptotic, and for large networks we establish the network immunity as the unique fixed point of an analytic function depending on the degree distributions and the initial seed size.
We then establish the asymptotic limit for the fraction of susceptible individuals. Moreover, we establish the basic reproduction number of the epidemic $\fR_0$ in the context of our model, 
defined as the expected number of  links of those individuals infected by one initial seed. We show that $\fR_0$ characterizes the spread of epidemic in the usual way. If this is larger than one, the epidemic starting from  few initial infectives  is explosive, whereas if it is below one then the epidemic will die out.

Our main results refer to the equilibrium of the social distancing game. 
All individuals choose their strategies and, in equilibrium, the realized network immunity -- determined using our first set of results -- must match the hypothesized network immunity.
The space of social activity levels is discretized, with zero representing (fully) social distancing and the maximum given by a government imposed level.
Individuals are assumed myopic: they make short term decisions, which have lifelong implications or even imply death.
For clarity, we think of the decision process as daily and of course, it suffices to scale our results for weekly, monthly or other short term horizon one deems realistic for the individuals' commitment to their social distancing strategy. 
Individuals derive short-term utility from social activity, and we assume that this scales linearly with their number of contacts. This is counteracted  by the probability of contracting the disease (over the same time horizon), multiplied by the type-specific death (or sequela)  probability given infection. 
The probability of contracting the disease clearly depends on the individual rates of contact and on the aggregate decisions of everyone else, through the network immunity.
We show that there is at most one equilibrium, which can be given semi-analytically.
For the case with two possible decisions, to socially distance or not, and when the graph is regular,  the social utility averaged over the population has a particularly simple form.
For the regular graph case, we show that the voluntary social distancing will always lead to a lower average utility than the social optimum, and this result holds irrespective of the functional dependence of the death rate on network immunity.
Put simply, even when people are in full recognition of the impact on the heath system and health outcomes of having a large outbreak, their decisions will lead to worse utility than in the social optimum. 

%We then proceed to examine numerically the gap between the Laissez-Faire equilibrium and social optimum for our model, calibrated to the Covid-19 current data. Several points emerge from this study: As we increase the fraction of social contact in VSLY, all age groups will practice less and less social distancing. However, for the youngest cohorts, the rate of decrease is highest. This effect could only increase if the fraction of social contact is non-constant across age groups and higher for the younger ones. Second, if individuals overestimate network immunity (or  the epidemic size is downplayed), then they will choose higher levels of activity than if they had perfect knowledge of the state of the network epidemic. In doing so, the epidemic becomes large.

%We next investigate the utility gap between the social optimum and the voluntary social distancing. We find that the gap is one magnitude more significant when the death rates given infection depend on the global  network immunity.

\paragraph{Related literature.}
Our work is part of the vast literature on  SIR epidemics on random networks, to name just a few~\cite{Ball16, ball2012sir,  bhamidi2019survival, HofstadLee16, Janson14, Pastor-Satorras15, ball2009threshold, kiss2017mathematics, ball2010analysis,ball2014epidemics,barbour2013approximating,britton2007graphs,draief2010epidemics, volz2008sir, deijfen2011epidemics}.
We continue on the same line as \cite{Janson14} who study the SIR epidemics dynamics in the configuration model. 
Our model is quite different: \cite{Janson14} look at a dynamic contagion model but the model is homogeneous, whereas we look at a final outcome of the cascade in the heterogeneous case.
In particular, we allow for different individual types, social distancing strategy and heterogeneous infection rates. Our proof approach also is also different, based on a reduction to the independent threshold model. This  allows us to consider a more general class of epidemics represented by the independent threshold model with differentiated types. This may be of interest in itself. 

Following the health emergency, several papers study   the equilibrium social distancing for COVID-19, see e.g.,~\cite{Acemoglu20, acemoglu2020testing,  NBERw27059, NBERw26984, ferguson2006strategies, ferguson2020report, hota2020impacts, del2007mixing, prem2020effect, miller2010incidence, toxvaerd2020equilibrium}. Previous works on social distancing games include \cite{bhattacharyya2019game}.
{\blue Our work is also related to vaccination and isolation games in
\cite{hota2020impacts, omic2009protecting, trajanovski2015decentralized}.}

Their  assumption is that at each of a set of pre-specified discrete time points, an individual can choose to pay a cost associated with social distancing in exchange for a reduction in the risk of
infection. In contrast to our work, such a reduction is not endogenous, or rigorously determined in equilibrium.
Our paper is to our best knowledge the first to allow for the rigorous analysis of the social distancing game on a heterogeneous network underlying social contacts. In a companion paper \cite{aminca20}, we calibrate our model to the characteristics of the Covid-19 pandemic.

The second related strand of literature is on economics of information security for homogeneous networks,  see e.g.~\cite{Gordon2002, lelarge12a, Acemoglu16} and  on games on network \cite{jackson2010social, jackson2015games}.
In~\cite{Gordon2002}, the authors consider a simple one-period economic model for a single individual who takes into account a monetary loss should infection occur and a  probability depending on security investment to become infected.  The security investment choice is analogous to the social distancing. 
A sufficient condition for monotone investment is given in  \cite{lelarge12a} which guarantees that  when network vulnerability is higher individuals invest more. 
We make the equivalent assumption here that there is less social distancing when  network immunity is higher. 
In particular, \cite{Acemoglu16, lelarge12a}  analyze the network security game (strong versus weak protection) for a simple (SI) contact process in random networks. 
We generalize their results by considering a heterogeneous SIR epidemic process and allowing for different social distancing levels.

\paragraph{Outline.} The paper is structured as follows. In Section~\ref{sec:models} we provide the  modeling framework for heterogeneous SIR epidemics and state our main results regarding the final outcome of the epidemic on random networks underlying the social contacts. In Section~\ref{sec:socialEq} we consider the network social distancing game. 
%In Section~\ref{sec:num}, we illustrate how our model can be applied to the Covid-19 public health crisis and calibrate the parameters. Section~\ref{sec:conclusion} concludes and the 
Section~\ref{sec:proofs} and Section~\ref{sec:threshold} contain the proofs.

\section{Models and  final outcome of major outbreak}\label{sec:models}
In this section we introduce the epidemic model and state the main results, for a given individuals social distancing strategies profile across individuals, when contact takes place on random networks with given vertex degrees.
\subsection{Heterogeneous SIR epidemics}
We consider a heterogeneous stochastic SIR epidemic process  with a possibility of death, i.e., a SIRD (Susceptible $\to$ Infectious $\to$ Recovered/Dead) model, which is a Markovian model for spreading a disease or virus in a finite population. 

We will refer to  Recovered/Dead
as Removed, and that whilst the distinction between Recovered/Dead is not important for the results of
this section it will become key in Section~\ref{sec:socialEq}.

Our population is assumed to interact via a network $\cG$. The set of nodes $[n] := \{1,2, \dots, n\}$ represents individuals or households, and the edges represent (potential) connections between individuals.  Connections can stem from various sources, and the network is understood to aggregate all these sources.
{\blue The degree of node $i$ is denoted by $d_i \in \NN$, which represents the number of nodes adjacent to node $i$ (neighbours).}
Individuals susceptible to the epidemic may become infected through contact with other infected neighbors. 

The population is heterogeneous, individuals can be of different types (e.g., age, sex, blood type, etc.) in a certain countable (finite or infinite) type space $\cT$, large enough to classify all individuals to the available information. We use the notations ${\bf t} = (t_1, t_2, \dots, t_n)$ to denote the type profiles of all individuals.

Moreover, we consider a finite ensemble of social distancing strategies $\cS=\{0, 1, 2, \dots, K\}$, with $0$ representing complete isolation and a higher value of $s\in \cS$ representing higher connectivity. We use the notations ${\bf s} = (s_1, s_2, \dots, s_n)$ and  ${\bf s}_{-i} = (s_1, \dots, s_{i-1},s_{i+1}, \dots, s_n)$  to denote the social distancing  profiles of all individuals and all individuals other than $i$ respectively.

We assume that at time $0$, all individuals have only partial information about the network characteristics, the epidemic parameters and the initial conditions.  An individual of type $t$ will get utility $u_t^{(s)}$ by choosing social distancing strategy $s\in \cS$. The social distancing equilibrium will be discussed in Section~\ref{sec:socialEq}.

Let us denote by $n^{(s)}_t$ the number of individuals of type $t\in \cT$ with social distancing strategy $s\in \cS$  so that  $$\sum_{t\in \cT} \sum_{s\in \cS} n^{(s)}_t = n.$$ 
It is understood that the network is parametrized by its size (and indeed all quantities we define depend on $n$, which we leave out from the notation for simplicity). We seek to understand the outcomes of a major outbreak as the size of the network becomes large. The following  condition is standard: 
 for all $t\in \cT$ and $s\in \cS$,

\begin{equation}
\frac{n^{(s)}_t}{n} \longrightarrow \mu^{(s)}_t \ \ \text{as} \ \ n \to \infty.
\end{equation}
We  let $\mu_t:= \sum_{s\in \cS} \mu_t^{(s)}$ be the (asymptotic) fraction of individuals with type $t$.

\medskip

The initial condition of the epidemic is given by the set of initially infected individuals $\cI(0)$, the set of initially removed individuals $\cR(0)$  and the set of susceptible individuals $\cS(0)$. The set of initially removed individuals could be interpreted as a set of immune or non-susceptible individuals.
In the later stages of the epidemic, the set of removed nodes will grow with the recovered or dead individuals. 
Note that $ \cS(0)  \cup \cI(0)\cup \cR(0) = [n]$ and the initial conditions may also depend on the type of each individual. 

Each infected individual, throughout their infection period, infects (makes infectious contact with) any susceptible 
neighbor individual with type $t$ and social distancing strategy $s$ at the points of a Poisson process with rate $\beta_{t}^{(s)}>0$. 
We assume that there is no latent period so that the contacted susceptible individuals  are immediately infected and are able to infect other individuals.

To simplify the analysis, we assume that the infection period $\rho$ is constant and, without loss of generality, we scale the time to make the constant $\rho=1$. 
The infected individual with type $t$ dies after time $\rho$ with fatality probability $\kappa_t$ and  recovers with probability  $1-\kappa_t$. 
The fatality probability  is important for incorporating
costs into the model which are key for the game theory part.

Note that if  $\rho\neq 1$, it suffices to replace the infection rates $\beta$ by $\rho \beta$.

% where  $\kappa_t\in[0,1]$ is the infected fatality rate for an individual of type $t\in \cT$.
We will allow in Section~\ref{sec:socialEq} the fatality probability to depend  on the fraction (number) of infected people during the epidemic process, as the hospital system can be overwhelmed. We assume that the infection rate depends only on the type and strategy of the susceptible party. Implicit is a conservative setting in which the effort to avoid infection comes from the susceptibles. One could make additional assumptions on the infectives, on whom we could impose quarantine or we could model additional elements in their utility functions to entail concern for their family and friends. 
Here we leave these considerations aside, in order to focus on the individual's decision when the utility includes only her own value of life.

 We also assume that the recovered individuals  are no longer infectious, and moreover immune to further infections. Note that this remains a point of active research for  Covid-19.
 The epidemic process continues until there are no infective individuals present in the population. Each alive individual is then either still susceptible, or else they have been infected and have recovered.

We assume that there are initially $n_{S}, n_I$ and $n_R$ susceptible, infective and removed (recovered or dead) individuals, respectively.
Moreover, for each type $t \in \cT$, there are respectively $n_{S, t}, n_{I, t}$ and $n_{R, t}$ of these individuals with type $t$.
Hence, we have  $|\cS(0)|=n_S, |\cI(0)|=n_I, |\cR(0)|=n_R$, $$n_S=\sum_{t \in \cT}n_{S, t}, \ n_{I}=\sum_{t \in \cT} n_{I,t}, \ n_{R}=\sum_{t \in \cT} n_{R,t} \ \text{and} \ n_{S}+n_{I}+n_{R}=n.$$ 

We are then interested in $\cS^{(\bs)} (\infty)$ and $\cR^{(\bs)} (\infty)$ the final set of susceptible and removed individuals, respectively,  when the individuals follow the social distancing strategy $\bs$.
 Similarly, $\cS_{t,d}^{(s)}(\infty)$ denotes the final set of  susceptible individuals with type $t$, degree $d$ and social distancing strategy $s$.

\subsection{Random networks with fixed degrees}

For $n\in \NN$, let ${\bf d}^{(n)} = (d_i)_{i=1}^n$ be a sequence of
non-negative integers {\blue (representing the degree of each individual)} such that $\sum_{i=1}^n d_i$ is even. We now consider a configuration model for the underlying network.
We endow the set of  individuals $[n] := \{1,2, \dots, n\}$ with the sequence of degrees ${\bf d}^{(n)}$.
 We define a random
multigraph with given degree sequence $(d_i)_1^n$ as follows. To each node $i$, we associate $d_i$ labeled half-edges. All half-edges need to be paired to construct the graph, this is done by randomly matching them. When a half-edge of a node $i$ is paired
with a half-edge of a node $j$, we interpret this as an edge between $i$ and $j$.  We denote the resulting random graph by $\cG$
 and we write $(i,j) \in \cG$ for the event that there is an edge between $i$ and $j$. It is easy to see that conditional on the multigraph being simple graph, we obtain a uniformly distributed random graph wth these given degree equences; see e.g.~\cite[Proposition 7.15]{van2016random} and \cite[Theorem 3.1.12]{durrett07}.
 
We consider asymptotics as  $n\to \infty$ for the SIRD model on the configuration model. In the remainder of the paper we will use the notation $o_p$ and $\stackrel{p}{\longrightarrow}$  in a standard way. Let $\{ X_n \}_{n \in \mathbb{U}}$ be a sequence of real-valued random variables on a probability space
$ (\Omega, \mathbb{P})$.
If $c \in \mathbb{R}$ is a constant, we write $X_n \stackrel{p}{\longrightarrow} c$ to denote that $X_n$ converges in probability to $c$.
That is, for any $\epsilon >0$, we have $\mathbb{P} (|X_n - c|>\epsilon) \rightarrow 0$ as $n \rightarrow \infty$.
Let $\{ a_n \}_{n \in \mathbb{N}}$ be a sequence of real numbers that tends to infinity as $n \rightarrow \infty$.
We write $X_n = o_p (a_n)$, if $|X_n|/a_n$ converges to 0 in probability.
If $\mathcal{E}_n$ is a measurable subset of $\Omega$, for any $n \in \mathbb{N}$, we say that the sequence
$\{ \mathcal{E}_n \}_{n \in \mathbb{N}}$ occurs with high probability ({\bf w.h.p.}) if $\mathbb{P} (\mathcal{E}_n) = 1-o(1)$, as
$n\rightarrow \infty$.
 
 \medskip
 
 We assume that there are initially $n^{(s)}_{S, t,d}$ susceptible individuals with social distancing strategy $s\in \cS$,  type $t\in \cT$ and degree $d \in \NN$. Further, there are $n_{I,d}$ and $n_{R,d}$ infective and recovered individuals with degree $d \in \NN$, respectively. Hence, we have $$n_{S}=\sum_{s\in \cS}\sum_{t \in \cT}\sum_{d=0}^{\infty}n^{(s)}_{S,t,d}, \ n_{I}=\sum_{d=0}^{\infty}n_{I, d}, \ n_{R}=\sum_{d=0}^{\infty}n_{R,d},$$ and $n_{S}+n_{I}+n_{R}=n$.

For the initially infected or removed individuals we do not need to know their distribution across types, and  only their total  number of links (infected linkages)  matters for the epidemic dynamics.
 We only need to know that their initial fraction converges as the network becomes large.  Similarly, we need convergence of the fraction of infected linkages.
 The type, degree and social distancing strategy distribution only matters for the susceptible individuals.

We now describe the regularity assumptions on individual degrees under individuals type profile ${\bf t}^{(n)} = (t_1, t_2, \dots, t_n)$ and social distancing strategy profile ${\bf s}^{(n)} = (s_1, s_2, \dots, s_n)$  . 
%Note that for each $n \in \NN$, we have a degree sequence  ${\bf d}^{(n)}$ and security profile ${\bf s}^{(n)} = (t_1, t_2, \dots, t_n)$, but (to lighten notation) this dependency on $n$ will not be carried in the notation.
We assume that the sequence $({\bs, \bf t, \bf d})$ and the set of initially susceptible, infective and recovered individuals  satisfies the following regularity conditions:

\begin{itemize}
\item[$(C_1)$] The fractions of initially susceptible, infective and recovered vertices converge to some $\alpha_S, \alpha_I, \alpha_R \in [0,1]$, i.e., {\blue as $n\to \infty$,}
\begin{equation}
n_S/n \to \alpha_S, \ \ n_I/n \to \alpha_I, \ \ n_R/n \to \alpha_R.
\end{equation}
Moreover, $\alpha_S > 0$.

\item[($C_2$)] The degree, type and social distancing strategy of a randomly chosen susceptible individual converges to {\blue (as $n\to \infty$)}
\begin{equation}
n^{(s)}_{S,t,d}/n_{S} \to \mu^{(s)}_{t,d},
\end{equation}
for some probability distribution $\left(\mu_{t,d}^{(s)}\right)_{s\in \cS, t\in \cT, d \in \NN}$.  Moreover, this limiting distribution has a finite and positive mean
$$\mu_S:=\sum_{s\in \cS}\sum_{t \in \cT}\sum_{d=0}^{\infty}d\mu^{(s)}_{t,d} \in (0,\infty),$$
and the average degree of a randomly chosen susceptible individual converges to $\mu_S$, i.e., {\blue as $n\to \infty$,}
\begin{equation}
\sum_{s\in \cS}\sum_{t \in \cT}\sum_{d=0}^{\infty}d n^{(s)}_{S,t,d}/n_S \to \mu_S.
\end{equation}

\item[($C_3$)] The average degree over all individuals converges to some $\lambda \in (0,\infty)$, i.e. as $n\to \infty$ 
\begin{equation}
\frac{1}{n}\sum _{i=1}^{n} d_i \to \lambda,
\end{equation}
and, in more detail, for some $\lambda_S, \lambda_I, \lambda_R$, the average degrees over susceptible, infective and recovered individuals converge:
\begin{equation}
\sum_{s\in \cS}\sum_{t \in \cT}\sum_{d=0}^{\infty}d n^{(s)}_{S,t,d}/n \to \lambda_S, \ \ \sum_{d=0}^{\infty}d n_{I,d}/n \to \lambda_I, \ \ \sum_{d=0}^{\infty}d n_{R,d}/n \to \lambda_R.
\end{equation}

\item[($C_4$)]The maximum degree of all individuals is not too large: 
\begin{equation}
d_{\max}=\max\{d_i: i=1, \dots, n\} = o(n).
\end{equation}
\end{itemize}

Our first theorem concerns the case where $\lambda_I>0$ and the initially infective population is macroscopic:

\begin{theorem}
\label{thm-main}
Suppose that $(C_1)-(C_4)$ hold and $\lambda_I>0$.
Then there is a unique solution $x_*^{(\bs)} \in (0,1)$ to the fixed point equation $x=f^{(\bs)}(x)$, where
\begin{equation}
f^{(\bs)}(x) := \frac{\lambda_R}{\lambda} + \alpha_S \sum\limits_{s\in \cS} \sum\limits_{t\in \cT}\sum\limits_{d=0}^{\infty}  \frac{d\mu^{(s)}_{t,d}}{\lambda} \left(x+(1-x)e^{-\beta_t^{(s)}
}\right)^{d-1} .
\label{eq:fixedpt}
\end{equation} 
Moreover, the final fraction of susceptible nodes with degree $d \in \NN$, type $t\in \cT$ and social distancing strategy $s\in \cS$ satisfies (as $n \to \infty$):
\begin{equation}
\frac{|\cS_{t,d}^{(s)}(\infty)|}{n_{S,t,d}^{(s)}} \stackrel{p}{\longrightarrow} \left( x_*^{(\bs)}+\big(1-x_*^{(\bs)}\big)e^{-\beta
_t^{(s)}}\right)^d.
\label{eq:suscepfraction}
\end{equation}

\end{theorem}

We can interpret  $x_*^{(\bs)}$ as the probability that a neighbor of a randomly chosen susceptible individual does not get infected during the epidemic. 
{\blue 
It is then obtained as a fixed point of Equation~\eqref{eq:fixedpt}, where
$f^{(\bs)}(x)$ can be interpreted as the probability that a neighbor of a randomly chosen susceptible individual does not get infected during the epidemic, if the second order neighbors are independently not getting infected with probability $x$. 
In order for any susceptible individual to not be infected, it must be all its neighbours are either not infected (with probability $x$) or are infected but did not make contact during the infection (with probability $(1-x)e^{-\beta_t^{(s)}}$).
 If the status of the neighbors are independent, and there are $d-1$ such neighbours then $\left(x+(1-x)e^{-\beta_t^{(s)}}
\right)^{d-1}$ represents the probability not to be infected.

Now consider not any node, but the neighbour of a randomly chosen susceptible individual. Either that neighbor is recovered, with probability 
given by the first term of Equation~\eqref{eq:fixedpt} which is  $\frac{\lambda_R}{\lambda}$, or she is susceptible. 
 In the latter case she has degree $d$ and strategy $s$ with
 size biased probability $    \frac{d n^{(s)}_{S,t,d}  }{\lambda_S n} $  which converges to  $\alpha_S \frac{d\mu^{(s)}_{t,d}}{\lambda}$ as $n$ goes to infinity.
 Indeed, we use size biased probability and not $\mu^{(s)}_{t,d}$ because we are considering a neighbour of a node and the individuals with degree $d$ are $d$ times more likely to be such neighbors (the likelihood scales linearly with $d$). 
 We conclude that, in the limit, the probability that a neighbor of a randomly chosen susceptible individual does not get infected during the epidemic is a fixed point solution $x_*^{(\bs)}$ of the equation $x=f^{(\bs)}(x)$. 

The same intuition applies to \eqref{eq:suscepfraction}. To be consistent with the susceptible status of an individual of degree $d$, type $t$ and strategy $s$, it must be that all its $d$ neighbors are either susceptible (with probability $x$) or they were removed before interaction (with probability $(1-x_*^{(\bs)}) e^{-\beta
_t^{(s)}} $ ). Hence $\left( x_*^{(\bs)}+\big(1-x_*^{(\bs)}\big)e^{-\beta
_t^{(s)}}\right)^d$ represents the probability that a randomly chosen individual with type $t$ and degree $d$ not to be infected (in the limit).

}

%{\blue {\bf Remark.} We set $\theta=\frac{\rho+\beta x}{\rho + \beta}$ in \cite{Janson14}} \dots \medskip

\medskip

Our next theorem concerns the case with initially few infective individuals, i.e. $|\cI(0)|=o(n)$ and $\lambda_I=0$. 
This result partly generalizes the result in \cite{Janson14} to the multi-type case with different infection rates depending on social distancing (as discussed in Section \ref{sec:intro}, their setup is dynamic and homogeneous).

Let
\begin{equation}
\fR_0^{(\bs)}:= \left(\frac{\alpha_{S}}{\lambda}\right) \sum\limits_{s\in \cS} \sum\limits_{t\in \cT}\left(1-e^{-\beta
_t^{(s)}}\right)  \sum\limits_{d=0}^{\infty}  d(d-1) \mu^{(s)}_{t,d}.
\end{equation}
%
%\begin{equation}
%\fR_0^{(\bs)}:= \left(\frac{\alpha_{S}}{\lambda}\right) \left(1-e^{-\bar{\beta}}\right)  \sum\limits_{s\in \cS} \sum\limits_{t\in \cT} \sum\limits_{d=0}^{\infty}  d(d-1) \mu^{(s)}_{t,d},
%\end{equation}
%where 
%$$e^{-\bar{\beta}}:=  \frac{1}{\lambda}\sum\limits_{s\in \cS} \sum\limits_{t\in \cT} \sum\limits_{d=0}^{\infty}  d \mu^{(s)}_{t,d}e^{-\beta
%_t^{(s)}}.$$

This quantity represents the expected number of infective links in the second generation of the epidemic, i.e., the number of linkages of  those infected by the initial seed (other than the link from the initial seed). It is these linkages that could propagate the epidemic.
Susceptible individuals are reached by the initial seed according to the size biased distribution $ \alpha_S \frac{d\mu^{(s)}_{t,d}}{\lambda} $ that we have seen above, and they are infected with probability $1-e^{-\beta
_t^{(s)}}$. 
In fact, when the initial fraction of susceptible individuals is macroscopic,  $\fR_0^{(\bs)}$ characterizes not only the second generation of the epidemic, but every generation in the initial stages of the epidemic. Initially, the contagion process behaves  like a branching process, which could either die out or explode. The following theorem states that if $\fR_0^{(\bs)}$ is below $1$ then the epidemic will die out and otherwise a positive fraction of the population will become infected.

\begin{theorem}
\label{thm-res}
Suppose that $(C_1)-(C_4)$ hold and $\lambda_I=\alpha_I=0$. The followings hold:
\begin{itemize}
\item[(i)] If $\fR^{(\bs)}_0<1$, then the number of susceptible individuals that ever get infected is $o_p(n)$.

\item[(ii)] If $\fR^{(\bs)}_0>1$, , then  there exists $\delta>0$ such that at least $\delta n$ susceptible individuals get infected with probability bounded away from zero.

%There is a unique solution $x_*^{(\bs)} \in (0,1)$ to the fixed point equation $x=f^{(\bs)}(x)$. Moreover, staring from a randomly chosen infected individual, the probability that pandemic does not die out (reach a positive fraction of population) is  $$1-G_\mu\big(x_*^{(\bs)}\big)$$ where $G_\mu(x):=\sum\limits_{s\in \cS} \sum\limits_{t\in \cT}\sum\limits_{d=0}^{\infty} \mu_{t,d}^{(s)} x^d $ denotes the generating function for distribution $\mu$. 

%Moreover, the final fraction (probability) of susceptible nodes with degree $d \in \NN$, type $t\in \cT$ and social distancing $s\in \cS$ satisfies:
%\begin{equation}
%\frac{|\cS_{t,d}^{(s)}(\infty)|}{n_{S,t,d}^{(s)}} \stackrel{p}{\longrightarrow} \left(\frac{\rho + x_*^{(\bs)}\beta_t^{(s)}}{\rho+\beta_t^{(s)}}\right)^d.
%\end{equation} 
\end{itemize}

\end{theorem}

We end this section by the following remark. Our results in this paper are all stated for the random multigraph $\cG$. However, they could be transferred by conditioning on the multigraph being a simple graph (without loop and multiples edges). The resulting random graph, denoted by $\cG_*$, will be uniformly distributed among all graphs with the same degrees sequence. In order to transfer the results, we would need (see e.g.,~\cite{Janson14}) to assume that the degree distribution has a finite second moment, i.e. 
$$\sum_{s\in \cS}\sum_{t \in \cT}\sum_{d=0}^{\infty}d^2\mu^{(s)}_{t,d} \in (0,\infty),$$
which from~\cite{janson06} implies that the probability that $\cG$ is simple being bounded away from zero as $n\to \infty$. Moreover, we conjecture that all results hold even without the second moment assumption. Indeed \cite{bollobas2015old} have recently shown results for the size of the giant component in $\cG_*$ from the multigraph case without using the second moment assumption.
%; they prove that even with the small probability that the multigraph is simple, the error probabilities are even smaller. 

\subsection{Vaccination}
We now consider heterogeneous SIR epidemics with vaccination. 
We assume that
vaccination is done before the outbreak and renders individuals completely immune to the
disease, meaning vaccinated susceptibles can be treated as removed.
The model becomes a percolated random graph $\cG$ where we first generate the random graph $\cG$ (by the configuration model) and then vaccinate (remove) individuals at random.  Given a probability function $\omega:\cT \times \NN \to [0,1]$, let $\mathcal{G}^{(n)}_{\omega}$  denote the random graph obtained by randomly and independently deleting each individual of type $t\in \cT$ and degree $d\in \NN$ with probability $\omega_{t,d}$. In particular (as an example), for edge-wise vaccination, one vaccinates the end point of each susceptible half edge with some fixed probability $\nu$ independently of all the other half-edges.
Thus the probability that
a degree $d$ susceptible individual is vaccinated will be 
$\omega_{t,d}=1-(1-\nu)^d.$

Note that in the case where the social planner does not have information on the types and degrees, degree vaccination, where we vaccinate the nodes with highest degrees, is not possible. 
Edge-based vaccination is more beneficial compared to random vaccination (see e.g.~\cite{Ball16, Janson14}).

In general, the total number of vaccinated individuals  will satisfy
\begin{equation}
V/n_S \stackrel{p}{\longrightarrow} \sum_{s\in \cS}\sum_{t \in \cT}\sum_{d=0}^{\infty} \omega_{t,d} \mu_{t,d}^{(s)}.
\end{equation}

Since vaccinating an individual would be equivalent to changing its type from susceptible to recovered, our results apply to the number of infected individuals after vaccination:

\begin{theorem}
\label{thm-vacc}
Consider heterogeneous SIR epidemics in percolated (vaccinated) random graph $\cG$. Suppose that $(C_1)-(C_4)$ hold. Suppose further that each initially susceptible individual with type $t\in \cT$ and degree $k\in\NN$ is vaccinated  with probability $\omega_{t,d}\in[0,1)$ independently of the others, and let $\lambda_I=\alpha_I=0$. Let   
\begin{equation}
\fR_{\omega}^{(\bs)}:= \left(\frac{\alpha_{S}}{\lambda}\right) \sum\limits_{s\in \cS} \sum\limits_{t\in \cT}\left(1-e^{-\beta
_t^{(s)}}\right)  \sum\limits_{d=0}^{\infty}  d(d-1) \mu^{(s)}_{t,d}(1-\omega_{t,d}).
\end{equation}

The following hold:
\begin{itemize}
\item[(i)] If $\fR^{(\bs)}_\omega<1$, then the number of susceptible individuals that ever get infected is $o_p(n)$.

\item[(ii)] If $\fR^{(\bs)}_\omega>1$, then  there exists $\delta>0$ such that at least $\delta n$ susceptible individuals get infected with probability bounded away zero.

\end{itemize}

\end{theorem}

This theorem can be obtained as a corollary of theorem \ref{thm-res}. Indeed, it suffices to augment the set of removed nodes by the set of vaccinated nodes.
The assumptions $(C_1)-(C_4)$ hold with high probability with the new distribution for the susceptible nodes $\omega_{t,d} \bigl(1-\mu_{t,d}^{(s)}\bigr)$. {\blue 
As above,  $\fR^{(\bs)}_\omega$ characterizes  every generation in the initial stages of the epidemic, but now some individuals have been vaccinated. The epidemic will die out if this quantity is below $1$. Since  $\fR^{(\bs)}_\omega$ clearly exhibits the contribution of individuals as a function of types, strategy and degree, it offers a simple way to rank these contributions and design a vaccination policy of minimal cost that ensures that the epidemic dies out. 
 }

\section{Equilibrium social distancing}
\label{sec:socialEq}
In this section we introduce the social distancing game and analyse its equilibrium.
The game takes place over one period. At time zero agents make their decisions, and the period ends when the epidemic ends. For simplicity, we assume that all agents keep their decisions constant over the entire period.

\subsection{The model}

We now consider a network social distancing game in presence of an epidemic risk. We assume that individual $i$ can decide on social activity level $s \in \cS:=\{0,1, \dots, K\}$  for a payoff $\pi_i^{(s)}=\pi^{(s)}_{t_i,d_i}$, and faces a potential loss $\ell_i$ in case it becomes infected. Clearly, deciding in a higher social activity increases the payoff, i.e., $\pi_i^{(s)}$ is strictly increasing in $s$. However, 
\begin{equation}
0= \beta_t^{(0)}<\beta_t^{(1)} < \dots < \beta_t^{(K)}.
\end{equation}
Further, in our baseline model, the government itself might imposes a maximum level of activity $K_g \leq K$, so the activity space is $\cS_g \subseteq \cS$. For example, all effective strategies will be capped by a level $K_g$ and the strategy becomes $s \wedge K_g$. 

At time zero (the beginning of the period) individuals decide on their social activity level based on the following  private and public information. 
Their only private information is their own state and their potential loss in case they  become infected during the epidemic.
The distribution of the (type-dependent) potential losses, denoted by $F_t$, is public knowledge. 
Regarding the network, only statistical quantities are public knowledge. 
Namely, individuals do not observe who is connected to whom and not even their own neighbors, only the degree-type distribution is known.
The degree-type and epidemic parameters $\beta_{t}^{(s)}$ are common knowledge as well.
Individuals do not know the exact nodes that are initially infected, but only their (asymptotic) fraction.

We assume that individual $i$ with type $t_i \in \cT$ who become infected incur a loss (of say death or serious illness) with probability $\kappa_{t_i}$. This loss is denoted by $\ell_i$. 

We can now define the  utility of (susceptible) node $i$ (in the network of size $n$) as
\begin{equation}
 u_i ({\bf \ell}, {\bf s}) =  u_i(\ell_1, ..., \ell_n, s_1, ..., s_n)  :=  \pi_{t_i,d_i}^{(s_i)} - \ell_i \kappa_{t_i}\PP_n (i \in \iI^{(\bs)} (\infty)),
 \label{eq:utility}
\end{equation}
where $\cI^{(\bs)}(\infty) = \cR^{(\bs)}(\infty)\setminus\cR(0)$ denotes the final set of all cumulative infection starting from initial infected seed $\cI_0$ and $\PP_n (i \in \iI^{(\bs)}(\infty))$ is over the distribution of the random graph $\cG$ of size $n$, given all nodes' degrees, losses and social activity vector $\bs$. As in \cite{NBERw27059}, we capture  the risk of loosing life or becoming ill  together by a single function $\kappa_t(x)$, where $t$ is the type and $x$ is the network immunity parameter.

We say that a  social activity across susceptible individuals ${\bf s^*} = (s_1^*, s_2^*, \dots, s_n^*)$ is a (pure-strategy)  {\it Nash equilibrium} if
\begin{equation}
s_i^* \in \mathop{\arg \max}_{s \in \cS}  u_i(\ell_1, ..., \ell_n, s_1^*, ...s_{i-1}^*,s, s_{i+1}^*, ..., s_n^*),
\end{equation}
for all $i \in [n]$.

Similarly, a social activity profile ${\bf s^*} = (s_1^*, s_2^*, \dots, s_n^*)$ is a {\it social optimum} if for all ${\bf s} \in \cS^n$, 
$$\sum_{i=1}^n  u_i ({\bf \ell}, {\bf s^*}) \geq \sum_{i=1}^n  u_i ({\bf \ell}, {\bf s}).$$

 We call parameter $x_*$ the global  network immunity, because it represents the probability that a susceptible neighbor of a randomly chosen node does not get infected. Then, for a given random network with immunity $x$, a representative susceptible individual with type $t$, degree $d$ and social activity $s$ will get infected and face losses $\ell$ with (asymptotic) probability (see Theorem~\ref{thm-main})
$$\kappa_t(x)   \left(1-\left( x+(1-x)e^{-\beta
_t^{(s)}}\right)^d\right).$$   
An individual of type $t$ and degree $d$ will  get utility $u_{t,d}^{(s)}(x)$ by choosing social distancing strategy $s\in \cS$.  More precisely, her utility is decomposed into the utility from activity, denoted by $\pi_{t,d}^{(s)}$ and a cost of life loss or path to recovery modeled by a loss random variable $\ell$ (following distribution $F_t$) and fatality probability $\kappa_t$.  We also assume that $\kappa_t=\kappa_t(x_*)$ depends on network immunity $x_*$. This captures the fact that the chance of death/severe illness is impacted by the performance of the health system, which in turn depends on the network immunity and the size of infected population.  Note that the utility from social activity depends on the choice of the susceptible individual and also on the overall fraction of individuals choosing to interact.

Given the global network immunity $x\in[0,1]$, the (representative) individual with type $t$ and degree $d$ maximization problem is thus 

\begin{equation}
  s^*_{t,d}:=  \arg\max_{s \in \cS} \pi_{t,d}^{(s)}  - \ell \kappa_t(x) \left ( 1-  \left( x+(1-x)e^{-\beta
_t^{(s)}}\right)^d \right).
\end{equation}
{\blue We  characterize the strategy of the representative individual as a function of his potential loss (in case of infection). However, the individuals with the same type and degree are allowed to differ in loss. }
We can interpret the solution to the above maximization problem as an {\it asymptotic Nash equilibrium} when the global network immunity $x$ summarizes the impact of all individuals optimal social activity levels $s^*_{t,d}$. {\blue This can be related to mean field models, where the network immunity can be seen as the mean field, aggregating the strategy of all players.}  This  is a fixed point problem that will be described in the following.
Indeed, this is an asymptotic Nash equilibrium because under partial information the limit of $\PP_n (i \in \iI^{(\bs)} (\infty))$ in \eqref{eq:utility} depends on the strategies of all other players only through the global network immunity.
Therefore, the strategy of each individual will be given by the strategy of the representative individual of her degree and type.

In particular, given global network immunity $x$, the representative individual prefers the social activity level $s$ over higher social activity level $s+1$ if and only if   

\begin{equation}
\ell \kappa_t(x) \left ( \left( x+(1-x)e^{-\beta
_t^{(s+1)}}\right)^d- \left( x+(1-x)e^{-\beta
_t^{(s)}}\right)^d   \right) >  \pi_{t,d}^{(s)}- \pi_{t,d}^{(s+1)}.
\end{equation}

Let us define for $s=0, 1, \dots, K-1$, the threshold loss functions
\begin{equation}
\ell_{t,d}^{(s)}(x):= \frac{\pi_{t,d}^{(s+1)}- \pi_{t,d}^{(s)}}{\kappa_t(x)\left ( \left( x+(1-x)e^{-\beta
_t^{(s)}}\right)^d- \left( x+(1-x)e^{-\beta
_t^{(s+1)}}\right)^d   \right) }.
\end{equation}
Note that $\ell_{t,d}^{(s)}(x)>0$ for $x \in (0,1)$ since $\pi_{t,d}^{(s)}$ and $\beta_t^{(s)}$ are strictly increasing in $s$ and $x\in(0,1)$ makes the denominator positive. 

Hence, the representative individual prefers the social activity level $s$ over higher social activity level $s+1$ if and only if  $\ell > \ell_{t,d}^{(s)}(x)$. In other words, since the loss function $\ell$ follows distribution $F_t$, the fraction of susceptible individuals with type $t$ and degree $d$ that prefer social activity $s+1$ over $s$ is given by $F_t\left(\ell_{t,d}^{(s)}(x)\right)$.

\begin{itemize}
\item[$(A_1)$] We assume in the following that $\ell_{t,d}^{(s)}(x)$ is a decreasing function of $s$, i.e., for all $d\in \NN, t\in \cT, x\in[0,1]$,
\begin{align}
\ell_{t,d}^{(0)}(x) >  \ell_{t,d}^{(1)}(x)> \dots  > \ell_{t,d}^{(K-1)}(x).
\end{align}
\end{itemize}
Note that the above assumption is only needed if there are more than two social activity levels, i.e., $K > 2$. Under this condition, the optimal individual's social activity is of threshold-type: follow the social activity level $s$ if and only if $\ell \in (\ell_{t,d}^{(s)}(x),\ell_{t,d}^{(s-1)}(x)]$ (set $\ell_{t,d}^{(-1)}(x)=\infty$).

\begin{itemize}
\item[$(A_2)$] We assume in the following that $\ell_{t,d}^{(s)}(x)$ is an increasing continuous function of $x$, i.e.
$$\kappa_t(x)\left ( \left( x+(1-x)e^{-\beta
_t^{(s)}}\right)^d- \left( x+(1-x)e^{-\beta
_t^{(s+1)}}\right)^d   \right)$$
is continuous and (strictly) decreasing in $x$.
\end{itemize}
%Note that if $\kappa_t(x)=\kappa_t$, the above assumption is automatically verified. 

The first assumption states that the level of loss where individuals choose to socially isolate is higher than the level of loss where individuals choose activity level $1$, and so on for the higher activity levels. 
The second assumption states that the level of loss where individuals choose to socially isolate is higher when the network immunity is higher. Same holds for all levels of activity;  for a fixed social distancing strategy the loss is higher when the global network health is higher. In other words, the less the global network immunity, the more susceptible individuals follow social distancing.
%
%The individual problem is then
%\begin{equation}
%    \max_{s \in \cS(\gamma)} u_t^{(s)}  - \left ( 1- \left(\frac{\rho + x_*^{(\bs)}\beta_t^{(s)}}{\rho+\beta_t^{(s)}}\right)^d \right)  \kappa_t(x_*^{(\bs)}).
%\end{equation}

{\blue We end this section by the following remark. 
The above model is conceived as a one period model, where the social distancing decision is made at the beginning of the period. However, one period may represent a verity of horizons (e.g. one day, one week, one month etc.) The equilibrium (steady state) of the epidemic is understood as the end of the period. In this case, the result at $\infty$  represents the end of the period.
Note that the network specification is consistent, the number of links refer to the number of contacts over one period (e.g. one day, one week, etc.) 
The time of the epidemic is an 'interaction time' and the relation to calendar time is carefully calibrated.
It does make sense that for a one period model the decision is kept constant. This is also for tractability reasons.
However, it is a very interesting extension to allow for multi-period decision making. In this case, we would need to update the types, degrees and decisions after each period. The myopic optimal decision problem will remain tractable. 
}

\subsection{Asymptotic Nash equilibrium analysis}

We are now ready to describe the asymptotic Nash equilibrium as a fixed point problem. In the previous section we described the representative individuals' choice given the global network immunity.

Let $x_e$ denote the expected global immunity in the random network (expected ratio of healthy edges among all the edges). Hence, the representative individual with degree $d$ and type $t$ would  choose social activity level $s$ if and only if 
$$\ell^{(s)}_{t,d} (x_e)< \ell \leq \ell^{(s-1)}_{t,d} (x_e).$$

So the fraction of individuals with degree $d$, type $t$ and following social activity level $s=0, 1, \dots, K$ is $\bar{\gamma}^{(s)}_{t,d} = \gamma^{(s)}_{t,d} (x_e)$, where 
\begin{equation}
\bar{\gamma}^{(s)}_{t,d} = F_t\left(\ell^{(s-1)}_{t,d} (x_e)\right)-F_t\left(\ell^{(s)}_{t,d} (x_e)\right)
\end{equation}
and we set $F_t(\ell^{(-1)})=F_t(\infty)=1$. So we have
$$\mu^{(s)}_{t,d}(x_e) = \mu_{t,d} \bar{\gamma}^{(s)}_{t,d}.$$

On the other hand,  given the probability distributions $\bar{\gamma}:\cT \times \NN \to \mathcal{P}(\cS)$, following Theorem~\ref{thm-main}, a node $i$ with type $t$ and degree $d$ will choose social activity $s\in \cS$ as long as 
  $$\ell^{(s)}_{t,d}(x_*^{\bar{\gamma}})<\ell_i \leq  \ell^{(s-1)}_{t,d}(x_*^{\bar{\gamma}}),$$
where $x_*^{\bar{\gamma}}$ is the smallest fixed point in $(0,1)$ of $x=f^{\bar{\gamma}}(x)$, with
\begin{equation}
f^{\bar{\gamma}}(x) := \frac{\lambda_R}{\lambda} + \alpha_S \sum\limits_{s\in \cS} \sum\limits_{t\in \cT}\sum\limits_{d=0}^{\infty}  \frac{d}{\lambda}  \mu_{t,d} \bar{\gamma}^{(s)}_{t,d} \left(x+(1-x)e^{-\beta
_t^{(s)}}\right)^{d-1}.
\end{equation} 

Hence, the actual fraction of individuals with type $t$ and degree $d$ following social activity $s\in \cS$ is given by $\gamma^{(s)}_{t,d}=\gamma^{(s)}_{t,d}(x_*^{\bar{\gamma}})$ where
\begin{align}
\gamma^{(s)}_{t,d} = F_t\left(\ell^{(s-1)}_{t,d} (x_*^{\bar{\gamma}})\right)-F_t\left(\ell^{(s)}_{t,d} (x_*^{\bar{\gamma}})\right).
\end{align}

Following the above analysis, for any $z\in[0,1], t\in \cT, d \in \NN$ and $s\in \cS$, we define  
\begin{align}
\label{eq:fpnash1}
f^{\gamma(z)}(x):=& \frac{\lambda_R}{\lambda} + \alpha_S \sum\limits_{s\in \cS} \sum\limits_{t\in \cT}\sum\limits_{d=0}^{\infty}  \frac{d}{\lambda}  \mu_{t,d} \gamma^{(s)}_{t,d}(z)  \left(x+(1-x)e^{-\beta
_t^{(s)}}\right)^{d-1},
\end{align}
where we set 
\begin{align}
\label{eq:fpnash2}
 \gamma^{(s)}_{t,d}(z) = F_t\left(\ell^{(s-1)}_{t,d} (z)\right)-F_t\left(\ell^{(s)}_{t,d} (z)\right).
 \end{align}
 
In the following theorem, we show the existence and uniqueness of (asymptotic) Nash equilibrium for the social distancing game.
\begin{theorem}\label{thm-unique}
Consider a random graph $\cG_n$ satisfying $(C_1)-(C_4)$. Assume that $(A_1)-
(A_2)$ holds. There exists a unique asymptotic Nash equilibrium (in the above sense), which is given by the solution of the following equation:
\begin{align}\label{eq:meanfield}
z= \inf\limits_{x \in [0,1]} \{x: f^{\gamma(z)}(x)=x \}.
\end{align}
\end{theorem}
The proof of above theorem is provided in Section~\ref{sec:proof-unique}.

Assuming that $n$ is large but finite, and all the players apply the asymptotic Nash equilibrium actions, it would be interesting to study the approximation of a Nash equilibrium in the finite network. We leave this and some other related issues for a future work.

\subsection{The effect of isolation}

In this section, we  assume $\mathcal S = \{0, 1\}$ and consider the (extreme) case where a susceptible individual following social distancing $s=0$ isolates and cannot get infected at all, i.e., $\beta^{(0)}_{t} = 0$ and $\beta^{(1)}_{t}= \beta_t$  for all $t\in \cT$. Namely, 
$s_i = 1$ if node $i$ does not quarantine and $s_i = 0$ if node $i$ quarantines and isolate from the network. Hence, given the network (expected) global immunity $x$ and setting $\pi_{t,d} = \pi_{t,d}^{(1)}-\pi_{t,d}^{(0)}$, a susceptible individual $i$ with type $t$ and degree $d$ will quarantine (isolate) from the network if and only if   
\begin{equation}
\ell_i> \ell_{t,d}(x):=\frac{\pi_{t,d}}{\kappa_t(x)\left (1- \left( x+(1-x)e^{-\beta
_t}\right)^d   \right) }.
\end{equation}

In this case, all individuals following the quarantine can be removed from the network and the epidemic goes through all other individuals. This is also equivalent to the individual vaccination (site percolation) model.  Let $\gamma_{t,d}$ denote the fraction of susceptible individuals (in equilibrium) with type $t$ and degree $d$ following quarantine. 

Hence, in this case $(A_1)-(A_2)$ are automatically verified and the equilibrium fixed point equations \eqref{eq:fpnash1}-\eqref{eq:fpnash2} can be simplified to:
\begin{align}
\gamma_{t,d} = 1-F_t\Bigl(\frac{\pi_{t,d}}{\kappa_t(x
_*^\gamma)\left (1- \left( x_*^\gamma+(1-x_*^\gamma)e^{-\beta
_t}\right)^d   \right) }\Bigr),
\end{align}
where 
$x_*^{\gamma}$ is the unique fixed point in $[0,1]$ of equation
\begin{equation}
f^{\gamma}(x) := \frac{\lambda_R}{\lambda} + \alpha_S \sum_{t\in \cT}\sum\limits_{d=0}^{\infty}  \frac{d}{\lambda}  \mu_{t,d} \left[\gamma_{t,d}+(1-\gamma_{t,d})  \left(x+(1-x)e^{-\beta
_t}\right)^{d-1}\right] .
\end{equation} 

%In Section \ref{sec:num} we will investigate the solution to this equation and give the application to Covid-19 for various network parameters. 
Let us assume $\pi_{t,d}^{(0)}=0$ (``no joy in isolation").  In this case, the social utility averaged over the population  converges to
\begin{align}
\frac 1 n \sum_{i=1}^n  u_i ({\bf \ell}, {\bf s}) \stackrel{p}{\longrightarrow} \bar{u}_{\rm social}(\gamma):=\sum_{t\in \cT}\sum\limits_{d=0}^{\infty}  \mu_{t,d}\bar{u}_{t,d}(\gamma),  \nonumber
\end{align}
with
\begin{equation}
\bar{u}_{t,d}(\gamma):= \pi_{t,d}(1-\gamma_{t,d}) -\kappa_t(x_*^\gamma)\left (1- \left( x_*^\gamma+(1-x_*^\gamma)e^{-\beta_t}\right)^d   \right) \int_\gamma^1 F_t^{-1}(1-u)du
\label{eq:socialoptbytype}
\end{equation}
where, for individuals with type $t$ and degree $d$, $\pi_{t,d}(1-\gamma_{t,d})$ is the total payoff from social activity  and $\kappa_t(x_*^\gamma) \int_\gamma^1 F_t^{-1}(1-u)du$ is the average  loss faced  by the $(1-\gamma_{t,d})$-fraction of individuals  who are not following isolation and therefore are subject to epidemic risk.

\subsection{Social optimum for regular homogeneous networks}
In this section, we consider the previous isolation setup in the case of random regular graphs, where $d_i=d$ and $t_i=t$ for all nodes $i\in [n]$. Hence, $\mu_{t,d}=1, \lambda=d, \lambda_R=\alpha_R d$, and we simplify the notations to $\kappa_t (\cdot)=\kappa(\cdot), \beta_t=\beta, \pi_{t,d}=\pi$ and $F_t=F$. 
Let $L$ be a random variable with distribution $F$. 

We use the following inverse cdf notation  for the loss distribution  
$$ F^{-1}(1-\gamma) := -\inf \{ \ell \in \RR : F(\ell) > 1 - \gamma\},$$ which represents the minimum amount of loss in $100(1-\gamma)\%$ worst-case scenarios. The equilibrium fixed point equations are simplified to 
\begin{align}
F^{-1}(1-\gamma) = \frac{\pi}{\kappa(x_*^\gamma)\left (1- \left( x_*^\gamma+(1-x_*^\gamma)e^{-\beta}\right)^d   \right)},
\end{align}
where 
$x_*^{\gamma}$ is the smallest fixed point in $[0,1]$ of equation
\begin{align}
f^{\gamma}(x) :=& \alpha_R+ \alpha_S \left[\gamma+ (1-\gamma) \left(x+(1-x)e^{-\beta}\right)^{d-1}\right] . 
%=&(1-\gamma) \left(\frac{\rho + x\beta}{\rho+\beta}\right)^{d-1}\\
%=& \alpha_R+ \alpha_S \left(\frac{\rho + x\beta}{\rho+\beta}\right)^{d-1} F\left(\frac{\pi\left(\rho+\beta\right)^d}{\kappa(x) \left ( \left(\rho+\beta\right)^d - \left(\rho+x_*^\gamma\beta\right)^d   \right)}\right).
\end{align} 
Let us assume again $\pi^{(0)}=0$ and $\pi^{(1)}=\pi$. The social utility averaged over the population  converges to
\begin{align}
\frac 1 n \sum_{i=1}^n  u_i ({\bf \ell}, {\bf s}) \stackrel{p}{\longrightarrow} \bar{u}_{\rm s}(\gamma) :=& \pi(1-\gamma) -\kappa(x_*^\gamma)\left (1- \left( x_*^\gamma+(1-x_*^\gamma)e^{-\beta}\right)^d   \right) \int_\gamma^1 F^{-1}_{1-u}(L)du \nonumber
\end{align}
where $\pi(1-\gamma)$ is the total payoff from social activity and $\kappa(x_*^\gamma) \int_\gamma^1 F^{-1}_{1-u}(L)du$ is the average  loss faced  by the $(1-\gamma)$-fraction of individuals  who are not following isolation and therefore are subject to epidemic risk.

The following theorem compares the fraction of self-isolating individuals in the {\it voluntary social distancing} equilibrium and the optimum reached by a social planner. 
\begin{theorem}\label{thm-social-reg}
The social planner will choose a larger fraction $\gamma$ of individuals following isolation than the equilibrium for any fixed payoff $\pi$ and fatality rate function $\kappa$.
\end{theorem}

{\blue 
We now illustrate the results on the simple case of a regular homogeneous network, in which all individuals have the same type and degree. They differ in their potential loss, which is assumed to follow a half normal distribution.
We seek insights into the self isolating policy when individuals over and under-estimate the network immunity. Our results also allow us to investigate how epidemic quantities such as the reproductive number $\fR_0$ differ in the voluntary social distancing equilibrium versus the social optimum. 
}

Figure~\ref{fig:Rega} varies the link payoff $\pi$ (this derives the gain from social participation and in the numerical calibration will be taken to equal the value of statistical life).  As the network immunity fixed point solution increases, the final fraction of infected individuals decreases. 
This is intuitive: all else fixed, as the link payoff becomes larger, less people choose to socially distance and this results into a large scale epidemic.
Note that the fraction of individuals who self-isolate in equilibrium is smaller than in the social optimum as long as the link payoff is sufficiently small.
When the link payoff is too high, then the equilibrium and socially optimal strategy will coincide, and agree on not socially distancing.

  \begin{figure}[htp]
\centering
\subfloat[]{\includegraphics[width=12cm,height=5.5cm]{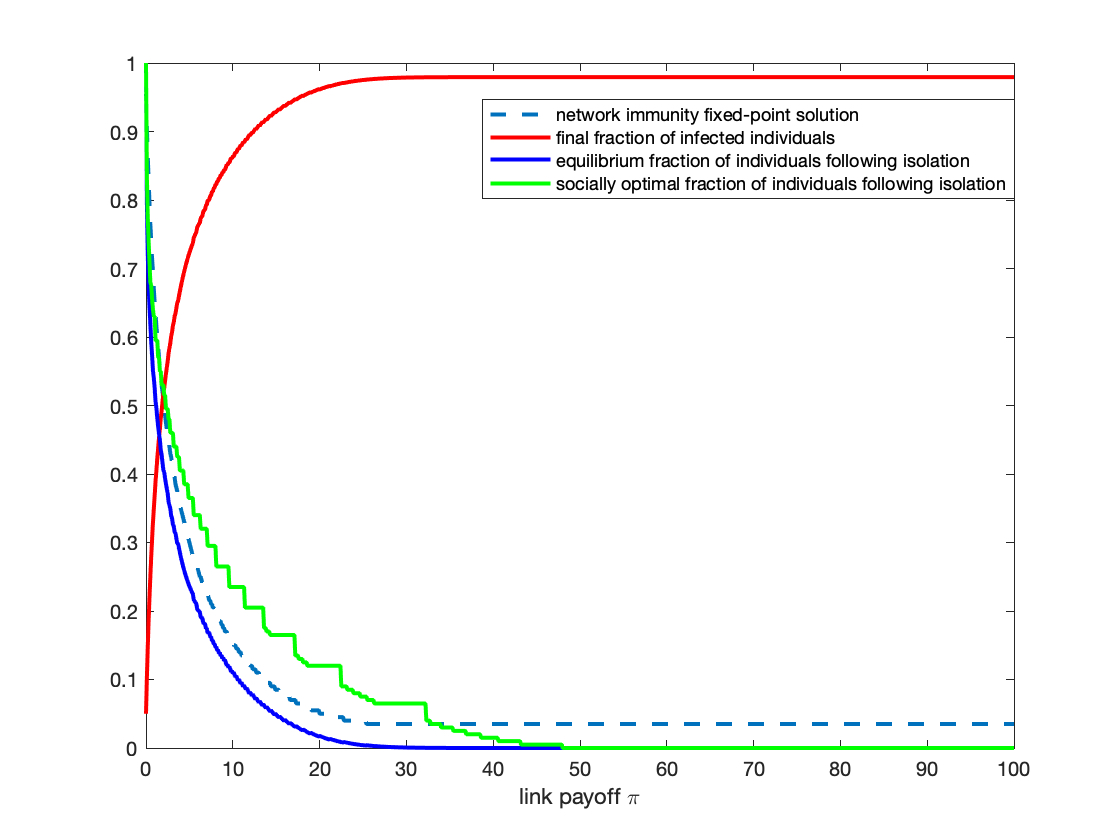}
\label{fig:Rega}}\\
\subfloat[]{\includegraphics[width=12cm,height=5.5cm]{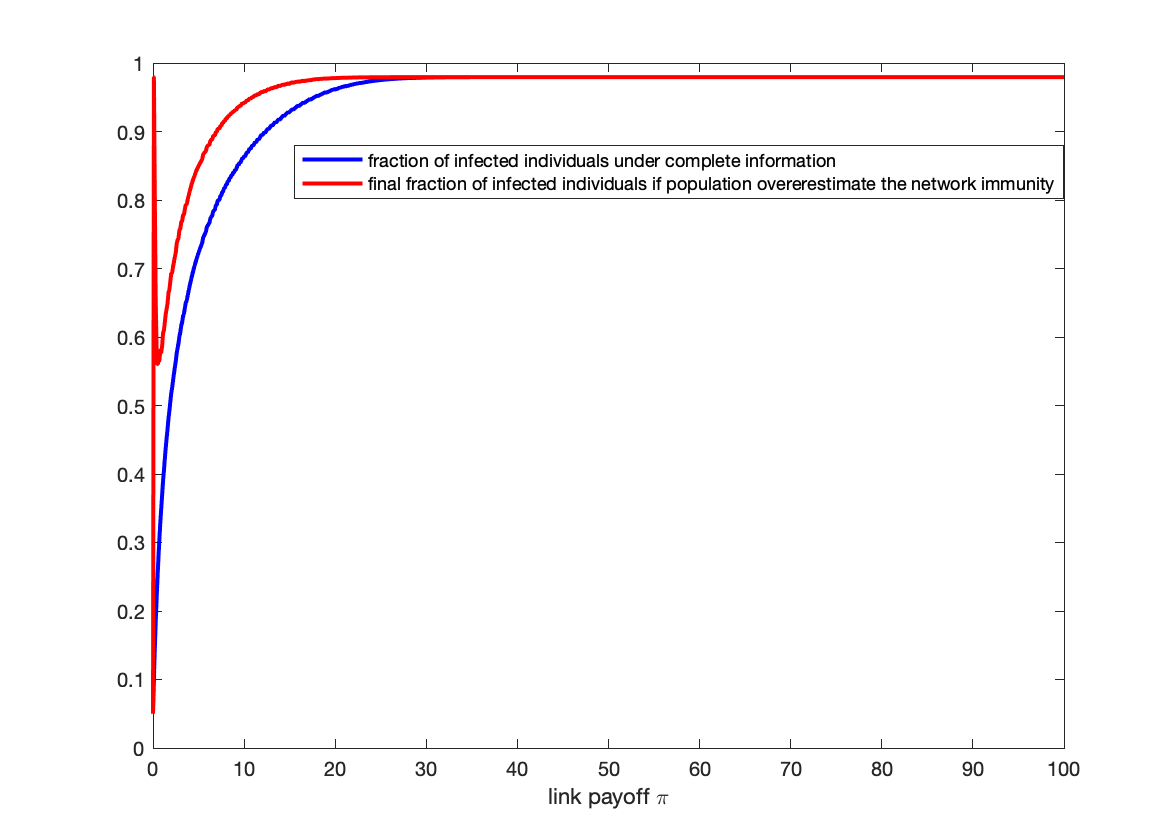}
\label{fig:Regb}}\\
\subfloat[]{\includegraphics[width=12cm,height=5.5cm]{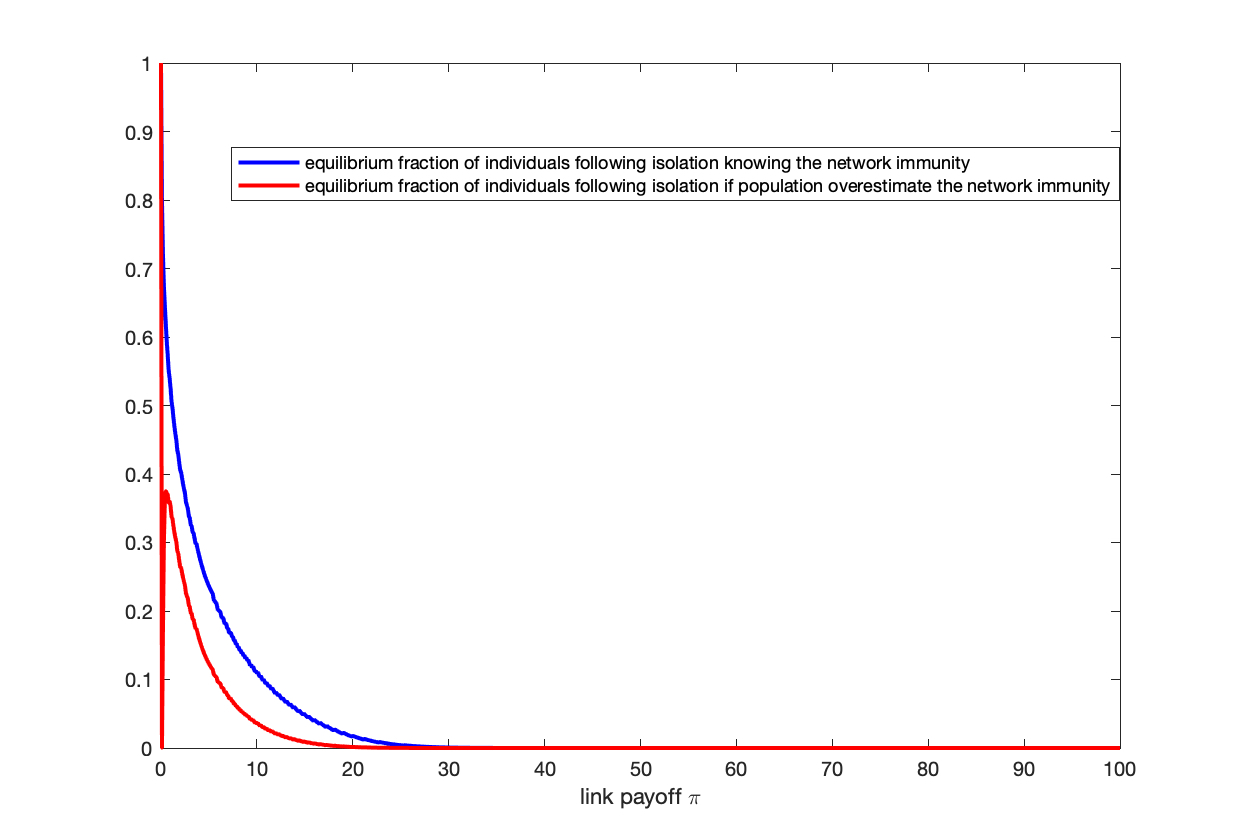}
\label{fig:Regc}}
\caption{\small Equilibrium solutions for regular homogeneous networks: Here $d=10, \alpha_R=0, \alpha_I=0.05, \alpha_S=0.95, \beta=0.4, \kappa(x)=0.1/(1+x)^3$ and $L$ follows half-normal distribution $L\sim HN(0,100)$ with density function $f(\ell;\sigma)=\frac{\sqrt{2}}{\sigma\sqrt{2\pi}}\exp\left(-\frac{\ell^2}{2\sigma^2}\right)$ for $\ell\geq 0$ and $\sigma=100$.}
%, i.e. $F(\ell)=1-e^{-\ell}$ and $F^{-1}_{1-\gamma}(L)=-\log(\gamma)$.

\label{fig:Reg}
\end{figure}
 
In Figures \ref{fig:Regb}-\ref{fig:Regc}
individuals may (1\%) overestimate the network immunity as they compute their optimal decision.
We plot the final fraction of infected individuals and individuals choosing to self isolate when they have perfect observation of the global network immunity versus when they overestimate the network immunity.
 Under over-estimation, a lower fraction self-isolate and the infection rates are higher.
The difference between the two cases can be seen as a value of information that allows individuals to optimally choose social distancing.
Even when payoffs from the linkages are low, an error on the estimated network risk can lead to a large scale epidemics. 
 Remark that  around $\pi = 0$, and when the immunity is equal to $1$, a small fraction of individuals can choose not to self isolate because their impact on overall immunity is small enough and with their over-estimation error, they  still believe that the network is fully immune.
Indeed, in the case when $\pi = 0$ and expected network immunity is one, because all individuals are indifferent, there are infinitely many equilibria.
This marks the beginning of the epidemic.

Figure \ref{fig:RegVac} shows that, as expected,  $\fR_0$ is larger in the voluntary social distancing equilibrium (Laissez-Faire equilibrium).  We note a strong dependence of the vaccination needs (in order to bring $\fR_0$ below $1$) on the link payoff.

\begin{figure}[htp]
\centering
\subfloat[]{\includegraphics[width=0.5\textwidth]{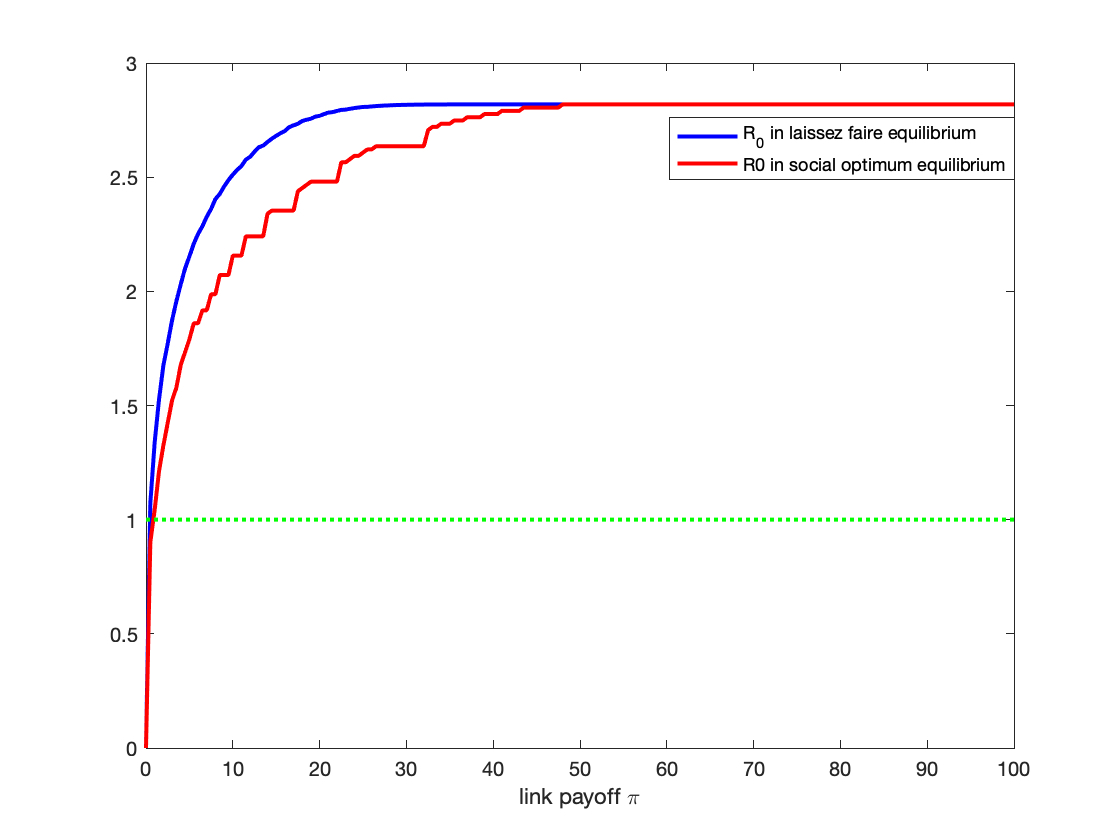}
\label{fig:RegVacA}}
\subfloat[]{\includegraphics[width=0.5\textwidth]{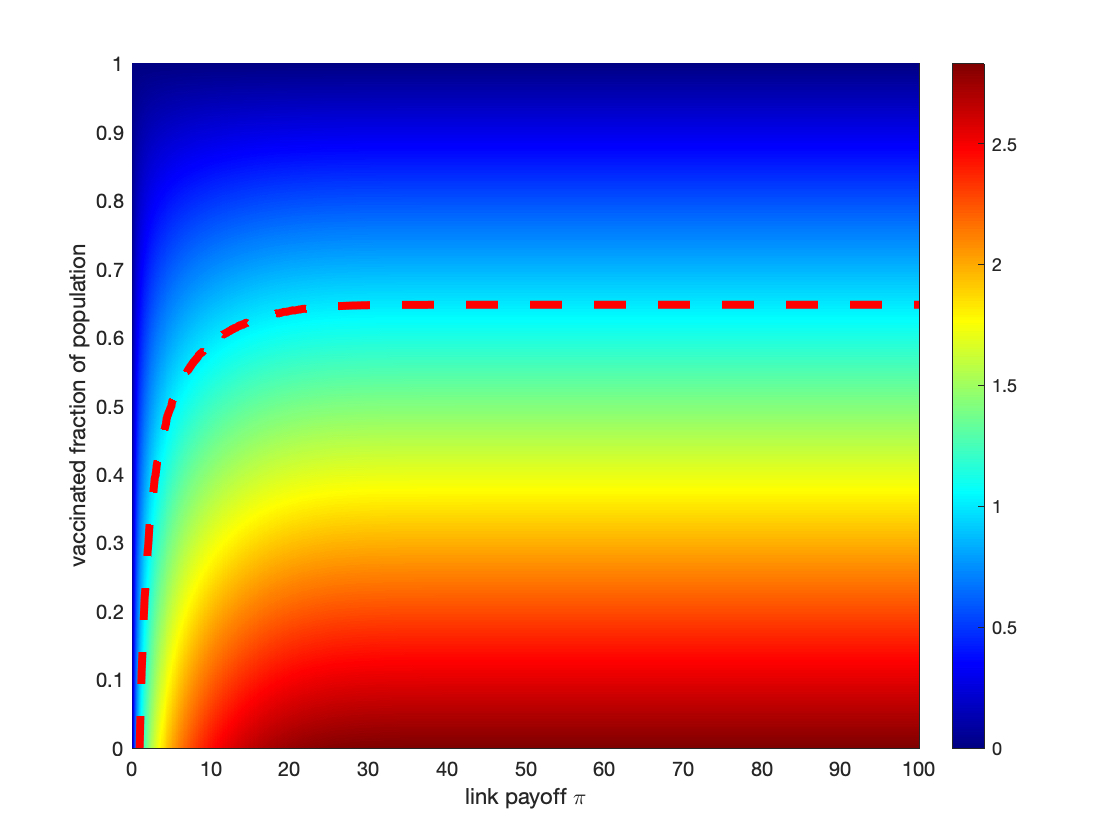}
\label{fig:RegVacB}}
\caption{\small (a) $\fR_0$ in voluntary social distancing equilibrium versus social optimum policy and (b) heat map for $\fR_0$ with vaccination in the case of regular homogeneous networks. Here $d=10, \alpha_R=0, \alpha_I=0.05, \alpha_S=0.95, \beta=0.4, \kappa(x)=0.1/(1+x)^3$ and $L$ follows half-normal distribution $L\sim HN(0,100)$.}
%, i.e. $F(\ell)=1-e^{-\ell}$ and $F^{-1}_{1-\gamma}(L)=-\log(\gamma)$.

\label{fig:RegVac}
\end{figure}

\medskip

We refer the reader to our companion paper \cite{aminca20}, where we calibrate our model to the Covid-19 epidemic characteristics and provide numerical results on the heterogeneous case.
We moreover study how different parameters affect the existence of an epidemic, its final size, and the utilities of the players.
{\blue 
The utility of the individuals integrates the value of statistical life:  social contact gives them a utility that is proportional to the value of statistical life yearly.  In case of infection there is a loss in utility, represented by the value of statistical life. There are important heterogeneity effects  both in the value of statistical life and in the  infection risk. When types represent age groups, the application of our model to the heterogeneous case can be used  to understand the social distancing strategies across age cohorts. We can then assess the  gap between the utility in the social distancing equilibrium and the social optimum across cohorts and derive implications of best type dependent policies.

}

%\newpage
\section{Proofs}
\label{sec:proofs}

\subsection{Proof of Theorem~\ref{thm-main}}
Consider the heterogeneous SIR epidemics spreading on $\cG$ satisfying $(C_1)-(C_4)$. In what follows, instead of taking a graph at random and then analyzing the epidemics, we use a standard coupling argument which allows us to study epidemics and the graph at the same time, revealing its edges dynamically while the epidemic spreads.
 
 Consider a vertex $i$ with type $t_i$ and $d_i$ (labelled) free (not yet paired) half-edges. We call a half-edge type $(s,t)$-susceptible, infective or removed according to the type of vertex it belongs to. A key step in the proof will be to decide from the beginning on the (random) infection threshold of each susceptible individual, denoted by $\Theta_i$ for individual $i$, defined as the (minimum) number of infected neighbor each individual can tolerate before it becomes infected. 
 
Since the (normalized) infection last $\rho=1$ days and the meeting happens at rate $\beta_t^{(s)}$ over all edges for a susceptible individual with type $t$, degree $d$ and following social activity $s$, it is easy to see that
\begin{equation}
\label{eq:threshold}
\PP(\Theta=\theta) = e^{-(\theta-1)\beta
_t^{(s)}}\left(1-e^{-\beta
_t^{(s)}}\right)=:p_{t,d}^{(s)}(\theta),
\end{equation}
 for $\theta=1,2, \dots, d$. Hence, $\mu_{t,d}^{(s)}p_{t,d}^{(s)} (\theta)$ will be the asymptotic fraction of susceptible individuals with type $d$, degree $d$, following social activity $s$ and getting infected after exactly $\theta$ infected neighbors.    We see that the model is equivalent to (type-dependent) independent threshold model for configuration model. 
 
The independent threshold model is defined by a given (type,degree,social distancing policy)-dependent threshold distribution $p_{t,d}^{(s)}(\theta)$. For each node with type $t\in \cT$, social distancing strategy $s\in \cS$ and degree $d\in \NN$, its threshold is drawn independently of everything else from this distribution. Starting from all initially infected nodes with threshold $0$, at each step, a susceptible node $i$ becomes infected if at least $\Theta_i$ neighbors are infected. 
 In the following, we first extend the results of \cite{amini, ACM11, lelarge11} on independent threshold model in configuration model, allowing for heterogeneous types and initial nodes removal. The theorem will then imply Theorem~\ref{thm-main}.
 
We denote by
$\Bin (k,p)$ a binomial  distribution  corresponding to the number of
successes of a sequence of $k$ independent Bernoulli trials each having probability of success  $p$.

\begin{theorem}\label{thm-independent}
Consider the type-dependent independent threshold model with threshold distribution $p_{t,d}^{(s)}(\theta)$ for all susceptible individuals with type $t$, degree $d$ and social distancing $s$, on random graph $\cG_n$ satisfying $(C_1)-(C_4)$. Let $x_*^{(\bs)}$ be the largest fixed point solution $x\in[0,1]$ to  $x=f^{(\bs)}(x)$ where 
\begin{equation}
f^{(\bs)}(x) := \frac{\lambda_R}{\lambda} + \alpha_S\sum\limits_{s\in \cS} \sum\limits_{t\in \cT} \sum\limits_{d=0}^{\infty}\sum\limits_{\theta=1}^d  \frac{d  \mu_{t,d}^{(s)}}{\lambda} p_{t,d}^{(s)}(\theta)  \PP\left(\Bin(d-1,1-x) \leq \theta-1 \right).
\end{equation}
We have for all $\epsilon>0$ w.h.p. 
\begin{align*}
\frac{|\cR^{(\bs)}(\infty)|}{n} \geq & \alpha_R+\alpha_S\sum\limits_{s\in \cS} \sum\limits_{t\in \cT} \sum\limits_{d=0}^{\infty}\sum\limits_{\theta=1}^d   \mu_{t,d}^{(s)} p_{t,d}^{(s)}(\theta)  \PP\left(\Bin(d,1-x_*^{(\bs)}) \geq \theta \right) -\epsilon.
\end{align*}
Moreover, if $x_*^{(\bs)}$ is a stable fixed point of $f^{(\bs)}(x)$, then 
\begin{align}
\frac{|\cR^{(\bs)}(\infty)|}{n} \stackrel{p}{\longrightarrow} \alpha_R+\alpha_S \sum\limits_{s\in \cS} \sum\limits_{t\in \cT} \sum\limits_{d=0}^{\infty}\sum\limits_{\theta=1}^d   \mu_{t,d}^{(s)} p_{t,d}^{(s)}(\theta)  \PP\left(\Bin(d,1-x_*^{(\bs)}) \geq \theta \right),
\end{align}
and,
the final fraction (probability) of susceptible nodes with degree $d \in \NN$, type $t\in \cT$ and social distancing strategy $s\in \cS$ satisfies:
\begin{equation}
\frac{|\cS_{t,d}^{(s)}(\infty)|}{n_{S,t,d}^{(s)}} \stackrel{p}{\longrightarrow}\sum\limits_{\theta=1}^d   p_{t,d}^{(s)}(\theta)  \PP\left(\Bin(d,1-x_*^{(\bs)}) \leq \theta-1 \right).
\end{equation} 
\end{theorem}

The proof of the above theorem is provided in Section~\ref{sec:threshold}. We now proceed with the proof of Theorem~\ref{thm-main} using the above theorem with 
$p_{t,d}^{(s)}(\theta)=e^{-(\theta-1)\beta
_t^{(s)}}\left(1-e^{-\beta
_t^{(s)}}\right)$,
 for $\theta=1,2, \dots, d$. In this case, using the binomial theorem, we have
 \begin{align*}
 \sum\limits_{\theta=1}^d  p_{t,d}^{(s)}(\theta)  \PP\left(\Bin(d-1,1-x) \leq \theta-1 \right) =&\sum\limits_{\theta=1}^d e^{-(\theta-1)\beta
_t^{(s)}}\left(1-e^{-\beta
_t^{(s)}}\right)  \PP\left(\Bin(d-1,1-x) \leq \theta-1 \right)\\
=&  \left(x+(1-x)e^{-\beta
_t^{(s)}}\right)^{d-1},
 \end{align*}
which implies that 
\begin{equation*}
f^{(\bs)}(x) = \frac{\lambda_R}{\lambda} + \alpha_S \sum\limits_{s\in \cS} \sum\limits_{t\in \cT}\sum\limits_{d=0}^{\infty}  \frac{d\mu^{(s)}_{t,d}}{\lambda} \left(x+(1-x)e^{-\beta
_t^{(s)}}\right)^{d-1} .
\end{equation*} 
as in Theorem~\ref{thm-main}. 
Moreover, the final fraction (probability) of susceptible nodes with degree $d \in \NN$, type $t\in \cT$ and social distancing strategy $s\in \cS$ satisfies:
\begin{align*}
\frac{|\cS_{t,d}^{(s)}(\infty)|}{n_{S,t,d}^{(s)}} \stackrel{p}{\longrightarrow}& \sum\limits_{\theta=1}^d   e^{-(\theta-1)\beta
_t^{(s)}}\left(1-e^{-\beta
_t^{(s)}}\right)  \PP\left(\Bin(d,1-x_*^{(\bs)}) \leq \theta-1 \right)\\
&= \left( x_*^{(\bs)}+\big(1-x_*^{(\bs)}\big)e^{-\beta
_t^{(s)}}\right)^d.
\end{align*} 

Hence, to prove Theorem~\ref{thm-main}, it only remains to prove that there is a unique solution $x_*^{(\bs)} \in (0,1)$ to the fixed point equation $x=f^{(\bs)}(x)$ which is a stable solution. Note that
$f^{(\bs)}(0)>0$ since $\alpha_S>0$ and $f^{(\bs)}(1)=\lambda_R/\lambda +\lambda_S/\lambda =1-\lambda_I<1$. Moreover, $f^{(\bs)}(x)$ is strictly increasing and continuous in $x$, which implies  there is a unique solution $x_*^{(\bs)} \in (0,1)$ to the fixed point equation $x=f^{(\bs)}(x)$. Moreover, this is a stable solution since $f^{(\bs)}(x)$ is strictly increasing.

\subsection{Proof of Theorem~\ref{thm-res}}
The proof of Theorem~\ref{thm-res} is based on Theorem~\ref{thm-main} and a theorem by Janson~\cite{janson08} on percolation in random graphs with given vertex degrees. Suppose that $(C_1)-(C_4)$ hold and $\lambda_I=\alpha_I=0$. We first show that if $\fR^{(\bs)}_0<1$, then the number of susceptible individuals that ever get infected is $o_p(n)$. We prove that in the subcritical case, if $\lambda_I=0$ then $x<f^{(\bs)}(x)$ for all $x\in[0,1)$. Indeed $f^{(\bs)}(1)=1$ (note that if $\lambda_I=0$, we have $\lambda=\lambda_R+\lambda_S$ which implies $f^{(\bs)}(1)=1$). Further,
\begin{align*}
\left(f^{(\bs)}(x)-x\right)'=& \alpha_S \sum\limits_{s\in \cS} \sum\limits_{t\in \cT}\sum\limits_{d=0}^{\infty} \left(1-e^{-\beta
_t^{(s)}}\right) \frac{d(d-1)\mu^{(s)}_{t,d}}{\lambda} \left(x+(1-x)e^{-\beta
_t^{(s)}}\right)^{d-2} -1 \\
\leq& \left(\frac{\alpha_{S}}{\lambda}\right) \sum\limits_{s\in \cS} \sum\limits_{t\in \cT}\left(1-e^{-\beta
_t^{(s)}}\right)  \sum\limits_{d=0}^{\infty}  d(d-1) \mu^{(s)}_{t,d} = \cR^{(\bs)}_0-1 <0.
\end{align*}
Since $x_I(\tau)=\lambda x \left(x-f^{(\bs)}(x)\right)$, we infer that 
$$\lim_{\alpha_I\to 0} x_*^{(\bs)} \to 1,$$
which implies that (by Theorem~\ref{thm-main}), the number of susceptible individuals that ever get infected is $o_p(n)$. 
 
We now consider the case $\fR_0^{(\bs)}>1$. Let us consider again the independent threshold model with threshold distribution
$$p_{t,d}^{(s)}(\theta):= e^{-(\theta-1)\beta
_t^{(s)}}\left(1-e^{-\beta
_t^{(s)}}\right).$$

Let only look at the structure of the subgraph obtained by removing all nodes with threshold higher than 1. Then each susceptible individual with type $t$, social activity $s$ and degree $d$ will remain in the percolated graph with probability 
$$p_{t,d}^{(s)}(1)= 1-e^{-\beta
_t^{(s)}}.$$

The result of Janson~\cite{janson08} on site percolation in configuration model implies that if 
$$\fR_0^{(\bs)}:= \left(\frac{\alpha_{S}}{\lambda}\right) \sum\limits_{s\in \cS} \sum\limits_{t\in \cT}p_{t,d}^{(s)}(1)  \sum\limits_{d=0}^{\infty}  d(d-1) \mu^{(s)}_{t,d} >1$$
then w.h.p. there is a giant connected component (where the fraction is bounded away from $0$) in the percolated random graph. Since all individuals present in the percolated random graph have threshold 1, the infection of any individual in the giant component will trigger the infection to whole component which implies Theorem~\ref{thm-res}.

\subsection{Proof of Theorem~\ref{thm-unique}}
\label{sec:proof-unique}
We define a function $g: [0,1] \to [0,1]$ via the following,
\begin{align*}
g(z) :=\inf\limits_{x \in [0,1]} \{x: f^{\gamma(z)}(x)=x \}.
\end{align*}
It can be easily seen that $f ^{\gamma(z)} (0) > 0, f^{\gamma(z)}(1)<1$. In conjunction with the continuity of $x \mapsto f ^{\gamma(z)}(x)$, we conclude that for any $z \in [0,1]$, the set $\{x: f^{\gamma(z)}(x)=x \}$ is nonempty and closed, and hence $g(z) \in (0,1)$ is well-defined. 

Now we show that $z \mapsto g(z)$ is decreasing in $z$, which implies that \eqref{eq:meanfield} has at most one solution. Suppose we have $0< z_1 < z_2 <1$. 

We have (set $F_t\left(\ell^{(-1)}_{t,d}\right) =1$ and )
\begin{align*}
f^{\gamma(z)}(x).=& \frac{\lambda_R}{\lambda} + \alpha_S \sum\limits_{s\in \cS} \sum\limits_{t\in \cT}\sum\limits_{d=0}^{\infty}  \frac{d}{\lambda}  \mu_{t,d} \gamma^{(s)}_{t,d}(z)\left(x+(1-x)e^{-\beta
_t^{(s)}}\right)^{d-1}\\
=&  \frac{\lambda_R}{\lambda} + \alpha_S  \sum\limits_{t\in \cT}\sum\limits_{d=0}^{\infty}  \frac{d}{\lambda}  \mu_{t,d} \sum_{s=0}^K \left(F_t\left(\ell^{(s-1)}_{t,d} (z)\right)-F_t\left(\ell^{(s)}_{t,d} (z)\right)\right)\left(x+(1-x)e^{-\beta
_t^{(s)}}\right)^{d-1} \\
=&  \frac{\lambda_R}{\lambda} + \alpha_S  \sum\limits_{t\in \cT}\sum\limits_{d=0}^{\infty}  \frac{d}{\lambda}  \mu_{t,d} \sum_{s=0}^K F_t\left(\ell^{(s-1)}_{t,d} (z)\right) \\
&\ \ \ \ \left(\left(x+(1-x)e^{-\beta
_t^{(s)}}\right)^{d-1} - \left(x+(1-x)e^{-\beta
_t^{(s-1)}}\right)^{d-1}\right).
\end{align*}

Hence, by using $(A_1)-(A_2)$, $f^{\gamma(z)}(x)$ is strictly increasing in $x$ and strictly decreasing function of $z$ (since $F$ is a strictly increasing cdf function). Therefore we have $f^{\gamma(z_1)}(x) > f^{\gamma(z_2)}(x)$ for any $x \in [0,1]$, and $$f^{\gamma(z_2)} ( g(z_1)) - g(z_1) < f^{\gamma(z_1)} (g(z_1)) -g(z_1) =0.$$ Combining with the fact that $f^{\gamma(z_2)}(0) \geq 0$ and the continuity of $x \mapsto f^{\gamma(z_2)}(x)$, there exists an $x < g(z_1)$ such that $ f^{\gamma(z_2)}(x)=x$, which implies that $g(z_2) < g(z_1)$. 

\subsection{Proof of Theorem~\ref{thm-social-reg}}

Recall that $\gamma_e$ (the equilibrium without social planner) is such that 
$$ \kappa(x)\left (1- \big( x_*^\gamma+(1-x_*^\gamma)e^{-\beta}\big)^d   \right) F^{-1}_{1-\gamma_e}(L)=\pi$$
while the social planner chooses $\gamma_s$ which maximizes $\bar{u}_{\rm social}(\gamma)$, i.e. 
$$\gamma_s = \mathop{\arg \max}_{\gamma \in [0,1]} \left\{\bar{u}_{\rm social}(\gamma):= \pi(1-\gamma) - \kappa(x_*^\gamma) \int_\gamma^1 F^{-1}_{1-u}(L)du\right\}.$$
Since $x_*^\gamma$ is increasing in $\gamma$ and $\kappa(.)$ is a decreasing function,  $\kappa(x_*^\gamma)$ is a decreasing function of $\gamma$ and we have
\begin{align*}
\bar{u}'_{\rm social}(\gamma_e) =& -\frac{d \kappa(x_*^\gamma)}{d\gamma} \int_\gamma^1 F^{-1}_{1-u}(L)du + \kappa(x_*^{\gamma_e}) F^{-1}_{1-\gamma_e}(L) - \pi \\
\geq& \kappa(x_*^{\gamma_e}) F^{-1}_{1-\gamma_e}(L) - \pi = \pi\left(\frac{1}{1- \big( x_*^\gamma+(1-x_*^\gamma)e^{-\beta}\big)^d  }-1\right) \geq 0,
\end{align*}
and the theorem follows.

\section{General independent threshold epidemics on $\cG$}
\label{sec:threshold}
In this section we present the proof of Theorem~\ref{thm-independent}.
\subsection{Markov chain transitions}\label{sec-markov}
We first describe the dynamics of the (independent threshold) epidemic on $\cG$ as a Markov chain, which is perfectly tailored for asymptotic study. At time $0$ the threshold of each susceptible individual is distributed randomly, according to (type dependent) distribution~\ref{eq:threshold}. 

For $\theta \in \NN$, let $n^{(s)}_{t,d,\theta}$ denotes the number of susceptible individuals with type $t\in \cT$, degree $d$ and social activity $s\in\cS$ which are given threshold $\theta=1, 2, \dots, d$. Hence, $$n^{(s)}_{t,d,\theta}/n_S \stackrel{p}{\longrightarrow} \mu_{t,d}^{(s)}p_{t,d}^{(s)} (\theta)$$ as $n\to \infty$. 
At a given time step $T$, individuals are partitioned into infected $\cI(T)$, susceptible $\cS(T)$ and removed $\cR(T)$. We further partition the class of susceptible nodes according to their type, social activity and threshold $$\cS(T) = \bigcup_{t, d, s, \theta}\cS_{t,d,\theta}^{(s)}(T).$$ 

At time zero, $\cI(0)$ and $\cR(0)$ contains respectively the initial set of infected and  recovered individuals. Hence, by ($C_1$), we know $|\cI(0)|/n\to \alpha_I$ and $|\cR(0)|/n\to \alpha_R$ as $n\to \infty$. 
 
 At each step we have one interaction only between two individuals, yielding at least one infected.
Our processes at each step is as follows :
 \begin{itemize}
 \item Choose a half-edge of an infected individual $i$;
 \item Identify its partner $j$ (i.e. by construction of the random graph in the configuration model, the partner is given by choosing a half-edge randomly among all available  half-edges);
 \item Delete both half-edges. If $j$ is currently uninfected with threshold $\theta$ and it is the $\theta$-th deleted half-edge from $j$, then $j$ becomes infected.
 \end{itemize}
 
Let us define $S^{(s)}_{t,d,\theta, \ell}(T)$, $0 \leq \ell < \theta$, the number of susceptible individuals with type $t$, degree $d$, social activity $s$, threshold $\theta$ and $\ell$ removed half-edges (infected neighbors)  at time $T$. We introduce the additional variables of interest:
\begin{itemize}
 \item $X_S(T)$: the number of (alive) susceptible half-edges belonging to susceptible individuals at time $T$;
 \item $X_I(T)$: the number of (alive) half-edges belonging to infected individuals at time $T$;
 \item $X_R(T)$: the number of (alive) half-edges belonging to initially recovered individuals at time $T$;
 \item $X(T)=X_S(T)+X_I(T)+X_R(T)$: the total number of (alive) half-edges at time $T$.
\end{itemize}
Hence, by Condition $(C_3)$, we have (as $n\to \infty$)
$$X_I(0)/n \longrightarrow \lambda_I,\ \  X_I(0)/n \longrightarrow \lambda_I, \ \ X_R(0)/n \longrightarrow \lambda_R \ \ \text{and} \ \ X(0)/n \longrightarrow \lambda.$$

Hence, $X(0)=\sum_{i=1}^n d_i$ denote the total number of half-edges in the network and,
since at each step we delete two half-edges, the number of existing (alive) half-edges at time $T$ will be 
\begin{equation}
X(T)=X(0)-2T.
\end{equation}
It is easy to see that the following identities hold:
\begin{align}
X_S(T)=& \sum\limits_{s\in \cS} \sum\limits_{t\in \cT} \sum\limits_{d=0}^{\infty}\sum_{\theta=1}^d\sum_{\ell=0}^{\theta-1} (d-\ell) S^{(s)}_{t,d,\theta, \ell}(T),\\
X_I(T)=& X(0)-2T-X_R(T)-X_S(T).
\end{align}

The contagion process will finish at the stopping time $T_*$ which is the first time $T\in \NN$ where $X_I(T)=0$. The final number of susceptible individual with type $t$, social distancing $s$, degree $d$ will be 
$$S^{(s)}_{t,d}(T_*) = \sum_{\theta=1}^\infty\sum_{\ell=0}^{\theta-1} S^{(s)}_{t,d,\theta, \ell}(T_*).$$
By definition of our process
$\mathbf{S}(T)=\left\{S^{(s)}_{t,d,\theta, \ell}(T)\right\}_{t,d,s,\theta, \ell}$ and $X_R(T)$ represent a Markov chain. We write the transition probabilities of the Markov chain. There are four possibilities for the $B$, the partner of a half-edge of an infected individual $A$:

\begin{enumerate}
\item $B$ is infected, the next state is $\mathbf{S}(T+1)=\mathbf{S}(T)$ and $X_R(T+1)=X_R(T)$;

\item $B$ is initially recovered. The probability of this event is $\frac{X_R(T)}{X(0)-2T}$. The changes for the next state will be $X_R(T+1)=X_R(T)-1$.

\item $B$ is uninfected of type $t$, degree $d$, social distancing strategy $s$, threshold $\theta$ and 
 this is the $(\ell+1)$-th deleted half-edge with $\ell+1 < \theta$. The probability of this event is $\frac{ (d-\ell) S^{(s)}_{t,d,\theta, \ell}(T)}{X(0)-2T}$. The changes for the next state will be
\begin{align*}
S^{(s)}_{t,d,\theta, \ell}(T+1) &= S^{(s)}_{t,d,\theta, \ell}(T) - 1 ,\\
S^{(s)}_{t,d,\theta, \ell+1}(T+1) &=S^{(s)}_{t,d,\theta, \ell}(T) + 1 .
\end{align*}

\item $B$ is uninfected of type $t$, degree $d$, social distancing strategy $s$, threshold $\theta$ and this is the $\theta$-th deleted incoming edge. The probability of this event is $\frac{ (d-\theta +1) S^{(s)}_{t,d,\theta, \ell}(T)}{X(0)-2T}$. The changes for the next state will be
\begin{align*}
S^{(s)}_{t,d,\theta, \theta-1}(T+1) &= S^{(s)}_{t,d,\theta, \theta-1}(T) - 1. 
\end{align*}
\end{enumerate}

Let $\Delta_T$ be the difference operator: $\Delta_T X := X(T+1)-X(T)$. We obtain the following equations for the expectation states variables, conditional on $\mathcal{F}_{T}$ (the pairing generated by time $T$), by averaging over the possible transitions:
\begin{eqnarray}
\EE\left[\Delta_T X_R |\mathcal{F}_{T} \right] &=& -\frac{X_R(T)}{X(0)-2T},\\
\EE\left[\Delta_T S^{(s)}_{t,d,\theta, 0} |\mathcal{F}_{T} \right] &=& - \frac{d S^{(s)}_{t,d,\theta, 0}(T)}{X(0)-2t} , \nonumber \\
\EE\left[\Delta_T S^{(s)}_{t,d,\theta, \ell} |\mathcal{F}_{T} \right] &=&\frac{(d-\ell+1) S^{(s)}_{t,d,\theta, \ell-1}}{X(0)-2t} - \frac{(d-\ell)S^{(s)}_{t,d,\theta, \ell}}{X(0)-2t}.
\label{eq:transitions}
\end{eqnarray}
The initial condition satisfies
\begin{eqnarray*}
X_R(0)/n \longrightarrow \lambda_R, \ \ \ S^{(s)}_{t,d,\theta, \ell}(0)/n \stackrel{p}{\longrightarrow} \alpha_S\mu_{t,d}^{(s)}p_{t,d}^{(s)} (\theta) \ind(\ell=0),
\end{eqnarray*}
as $n\to \infty$.
Remark that we are interested in the value of $S^{(s)}_{t,d,\theta, \ell}(T_*)$, where $T_*$ is the first time that $X_I(T_*)=0$. In case $T_* < X(0)$, the Markov chain can still be well defined for $t \in [T_*, X (0))$ by the same transition probabilities. However, after $T_*$ it will no longer be related to the epidemic process and the value $X_I(T)$, representing for $t \leq T_*$ the number of alive half-edges belonging to infected individuals, becomes negative. We consider from now on that the above transition probabilities hold for $ T < X(0)$.

\medskip

We will show next that the trajectory of these variables throughout the
algorithm is a.a.s. (asymptotically almost surely, as $n \rightarrow \infty$ ) close to the solution of the deterministic differential equations suggested by these equations.

\subsection{Fluid limit of the epidemic process}

 Consider the following system of differential equations (denoted by (DE)):

\begin{eqnarray*}
x_R'(\tau) &=& -\frac{x_R(\tau)}{\lambda-2\tau},\\
(s^{(s)}_{t,d,\theta, 0})'(\tau) &=& - \frac{d s^{(s)}_{t,d,\theta, 0}(\tau)}{\lambda-2\tau},\\
(s^{(s)}_{t,d,\theta, \ell})'(\tau) &=& \frac{(d-\ell+1) s^{(s)}_{t,d,\theta, \ell-1}(\tau)}{\lambda-2\tau} - \frac{(d-\ell) s^{(s)}_{t,d,\theta, \ell}}{\lambda-2\tau},\qquad ({\rm DE}),
\end{eqnarray*}
with initial conditions
\begin{eqnarray*}
x_R(0)=\lambda_R, \ \ s^{(s)}_{t,d,\theta, \ell}(0)&=& \alpha_S\mu_{t,d}^{(s)}p_{t,d}^{(s)} (\theta) \ind(\ell=0).
\end{eqnarray*}
\begin{lemma}\label{lem-sol}
The system of ordinary differential equations ({\rm DE}) admits the unique solution
$$x_R(\tau), \ \  \mathbf{s}(\tau) := \left\{s^{(s)}_{t,d,\theta, \ell}(\tau) \right\}_{s, t,d,\theta, \ell}$$
in the interval $0 \leq \tau < \lambda/2$, with
\begin{equation}
 x_R(\tau)=\lambda_R x, \ \  s^{(s)}_{t,d,\theta, \ell}(\tau) :=  \mu_{t,d}^{(s)}p_{t,d}^{(s)} (\theta) {d \choose \ell} x^{d-\ell}(1-x)^\ell,
 \label{solutions}
\end{equation}
where $x=\sqrt{1-2\tau/\lambda}$ and $0\leq \ell<\theta$.
\end{lemma}

\begin{proof}
Let $u=u(\tau)=-\frac 1 2 \ln (\lambda-2\tau)$. Note that $u(0)=-\frac{1}{2}\ln(\lambda)$, $u$ is strictly monotone and so is the inverse function $\tau=\tau(u)$. We write the system of differential equations with respect to $u$:
\begin{eqnarray*}
x_R'(u) &=& -x_R(u),\\
(s^{(s)}_{t,d,\theta, 0})'(u) &=& - d s^{(s)}_{t,d,\theta, 0}(u),\\
(s^{(s)}_{t,d,\theta, \ell})'(u) &=& (d-\ell+1) s^{(s)}_{t,d,\theta, \ell-1}(u)- (d-\ell) s^{(s)}_{t,d,\theta, \ell}(u).
\end{eqnarray*}

Then we have $$x_R(u)=\lambda_Re^{-(u-u(0))}=\frac{\lambda_R}{\lambda}\frac{\sqrt{\lambda-2\tau}}{\sqrt{\lambda}}=\lambda_R x,$$
$$\frac{d}{du} (s^{(s)}_{t,d,\theta, \ell+1}e^{(d-\ell-1)(u-u(0))}) = (d-\ell) s^{(s)}_{t,d,\theta, \ell}(u)
e^{(j-\ell-1)(\gamma-\gamma(0))},$$
and by induction, we find
\begin{eqnarray*}
s^{(s)}_{t,d,\theta, \ell}(u) = e^{-(d-\ell)(u-u(0))}  \sum_{r=0}^\ell {{d-r}\choose{\ell-r}} \left( 1 - e^{-(u-u(0))} \right)^{\ell-r} s^{(s)}_{t,d,\theta, r}(u(0)).
\end{eqnarray*}
By going back to $\tau$, we have
\begin{eqnarray*}
s^{(s)}_{t,d,\theta, \ell}= x^{d-\ell} \sum_{r=0}^\ell s^{(s)}_{t,d,\theta, r}(0) {{d-r}\choose{\ell-r}}(1-x)^{\ell-r}.
\end{eqnarray*}
Then, by using the initial conditions, we find (for $0\leq \ell<\theta$)
$$s^{(s)}_{t,d,\theta, \ell}(\tau) := \alpha_S \mu_{t,d}^{(s)}p_{t,d}^{(s)} (\theta) {d \choose \ell} x^{d-\ell}(1-x)^\ell.$$

\end{proof}

A key idea to prove Theorem~\ref{thm-main} is to approximate, following   \cite{Worm95}, the
  Markov chain by the solution of a system of differential equations in the large network limit. We summarize here the main result of \cite{Worm95}.

For a set of variables $x^1, ..., x^b$ and for $\cD \subseteq \RR^{b+1}$,
define the stopping time $$T_{\cD}=T_{\cD}(x^1, ..., x^b)=\inf\{t\geq 1, (t/n; x^1(t)/n, ..., x^b(t)/n) \notin \cD\}.$$
\begin{lemma}[\cite{Warnke2019, Worm95}]
\label{thm-eqdif1}
Given integers $b, n\geq 1$, a bounded domain $\cD\subseteq \RR^{b+1}$, functions  $\left(f_\ell\right)_{1\leq \ell \leq b}$ with $f_{\ell} : \cD \to \RR$, and $\sigma$-fields  $\mathcal{F}_{n,0}\subseteq \mathcal{F}_{n,1} \subseteq \dots$, suppose that the random variables $\left(Y_n^{\ell}(t)\right)_{1 \leq {\ell} \leq b}$ are  $\mathcal{F}_{n,t}$-measurable for $t\geq 0$. Furthermore, assume that, for all $0 \leq t < T_{\cD}$ and $1\leq \ell\leq b$, the following conditions hold

%Let $b\geq 2$ be an integer and consider a sequence of real valued random variables $(\{Y_n^{\ell}(t)\}_{1 \leq {\ell} \leq b})_{t\geq 0}$ and its natural filtration $\mathcal{F}_{n,t}$. Assume that there is a constant $C_0 > 0$ such that $|Y_n^{\ell}(t)| \leq C_0 n$ for all $n$, $t \geq 0$ and $1 \leq {\ell} \leq b$.
%For all $\ell \geq 1$ let $f_{\ell} : \RR^{b+1} \to \RR$ be functions and assume that for some bounded connected open set $\cD \subseteq \RR^{b+1}$ containing the
%% intersection of $\{(t,z_1,...,z_b): t\geq 0\}$ with some neighborhood of
%closure of
%$$\{(0,z_1,...,z_{b}):\exists \ n \mbox{ such that }\PP(\forall \ 1\leq {\ell} \leq b, \ Y_n^{\ell}(0)=z_{\ell} n) \neq 0 \},$$
%the following three conditions are verified:
\begin{itemize}
  \item[(i)] {\rm(Boundedness).} $\max_{1 \leq {\ell} \leq b} |Y_n^{\ell}(t+1)-Y_n^{\ell}(t)|\leq \beta,$
  \item[(ii)] {\rm(Trend-Lipschitz).} $|\EE[Y_n^{\ell}(t+1)-Y_n^{\ell}(t)|\mathcal{F}_{n,t}]-f_{\ell}(t/n,Y_n^1(t)/n,...,Y_n^{\ell}(t)/n)| \leq \delta$, where the function $(f_{\ell})$ is $L$-Lipschitz-continuous on $\cD$,
%  \item {\rm(Lipschitz).}  The functions $(f_{\ell})_{1 \leq {\ell} \leq b}$ are Lipschitz-continuous  on $\cD$.
\end{itemize}
and that the following condition holds initially:
\begin{itemize}
\item[(iii)] {\rm(Initial condition).} $\max_{1 \leq {\ell} \leq b} |Y_n^{\ell}(0)-\hat{y}^{\ell}n|\leq \alpha n,$ for some $\left(0,\hat{y}^1, \dots, \hat{y}^b\right) \in \cD$.
\end{itemize}
Then there are $R =R(\cD, L) \in [1, \infty)$ and $C=C(\cD)\in (0,\infty)$ such that, whenever $\alpha \geq \delta \min\{C, L^{-1}\}+R/n$, with probability at least  $1-2be^{-n\alpha^2/(8C\beta^2)}$ we have
$$\max_{0 \leq t \leq \sigma n}\max_{1\leq \ell \leq b}|Y_n^{\ell}(t) -  x^{\ell}(t/n)n| < 3 e^{CL}\alpha n,$$
where $\left(x^{\ell}(t)\right)_{1 \leq {\ell} \leq b}$ is the unique solution to the system of differential equations
$$\frac{dx^{\ell}(t)}{dt}=f_{\ell}(t, x^1,..., x^{b}) \ \ \text{with} \ \ x^\ell(0) = \hat{y}^\ell, \ \ \text{for} \ \ \ell=1,...,b ,$$
and $\sigma = \sigma(\hat{y}^1, \dots, \hat{y}^b) \in[0,C]$ is any choice of $\sigma\geq 0$ with the property that $(t, x^1(t),..., x^{b}(t))$ has $\ell^\infty$-distance at least $3e^{LC} \alpha$ from the boundary of $\cD$ for all $t\in [0, \sigma)$.
%Then the following conclusions hold:
%\begin{description}
%  \item[(a)] For $(0,\hat{z}_1,...,\hat{z}_{b}) \in D$, the system of differential equations
%  $$\frac{dz_{\ell}}{ds}=f_{\ell}(s,z_1,...,z_{\ell}), \ \ \ell=1,...,b ,$$ has a unique solution in $\cD$, $z_{\ell}:\RR \rightarrow \RR$, which passes through $z_{\ell}(0)=\hat{z}_{\ell}$, for $\ell=1,\dots,b$, and which extends to points arbitrarily close to the boundary of $\cD$.
%  \item[(b)] Let $\delta>\delta_1(n)$ with $\delta=o(1)$.
%  For a sufficiently large constant C, with probability $1-O\left(\frac{b\beta(n)}{\delta} \exp \left( - \frac{n\delta^3}{\beta(n)^3}\right)\right)$, we have
% $$\sup_{0 \leq t \leq \sigma(n) n}(Y_n^{\ell}(t) - n z_n^{\ell}(t/n)) = O(\delta n),$$ where $\mathbf{z}_n(t)=(z_n^1(t),\dots,z_n^b(t))$ is the solution of  $$ \frac{d\mathbf{z}_n}{dt}=f(t,\mathbf{z}_n(t))\qquad {z}_n(0) = \mathbf{Y}_n(0)/n$$  $${\rm and}\qquad\sigma(n)=\sup \{ t\geq 0,\quad d_{\infty}(\mathbf{z}_n(t),\partial D)\geq C\delta\}.$$
%\end{description}
\end{lemma}

We apply Lemma~\ref{thm-eqdif1} to the epidemic process described in Section \ref{sec-markov}.
Let us define, for $0 \leq \tau \leq \lambda/2$,
\begin{align}
x_S(\tau)=& \sum\limits_{s\in \cS} \sum\limits_{t\in \cT} \sum\limits_{d=0}^{\infty}\sum_{\theta=1}^d\sum_{\ell=0}^{\theta-1} (d-\ell) s^{(s)}_{t,d,\theta, \ell}(\tau),\\
x_I(\tau)=& \lambda-2\tau-x_R(\tau)-x_S(\tau). \label{eq:xI}
\end{align}
with $s^{(s)}_{t,d,\theta, \ell}$ and $x_R$ given in Lemma \ref{lem-sol}.
With $\Bin(d,x)$ denoting a binomial variable with parameters $d$ and $x$, we have
\begin{equation}
\label{eq:delt}
x_S(\tau) =   \alpha_S\sum\limits_{s\in \cS} \sum\limits_{t\in \cT} \sum\limits_{d=0}^{\infty}\sum_{\theta=0}^d   \mu_{t,d}^{(s)}p_{t,d}^{(s)} (\theta) (dx) \PP\left(\Bin(d-1,1-x) \leq \theta-1 \right),
\end{equation}
and, using $x=\sqrt{1-2\tau/\lambda}$ and Equation~\ref{eq:xI},
\begin{eqnarray*}
x_I(\tau)&=& \lambda-2\tau-\lambda_R x- \alpha_S\sum\limits_{s\in \cS} \sum\limits_{t\in \cT} \sum\limits_{d=0}^{\infty}\sum\limits_{\theta=1}^d   \mu_{t,d}^{(s)}p_{t,d}^{(s)}(\theta) (dx) \PP\left(\Bin(d-1,1-x) \leq \theta-1 \right)\\
&=&\lambda x\left(x-\frac{\lambda_R}{\lambda}- \alpha_S\sum\limits_{s\in \cS} \sum\limits_{t\in \cT} \sum\limits_{d=0}^{\infty}\sum\limits_{\theta=1}^d  \frac{d  \mu_{t,d}^{(s)}}{\lambda} p_{t,d}^{(s)}(\theta)  \PP\left(\Bin(d-1,1-x) \leq \theta-1 \right)  \right)\\
&=&(\lambda x) \left(x-f^{(\bs)}(x)\right).
\end{eqnarray*}

Since $x_*$ is the largest solution in $(0,1)$ to the fixed point equation $x=f^{(\bs)}(x)$, we have $x_*=\sqrt{1-2\tau_*/\lambda}$ where $\tau_* $ is the smallest $\tau \in (0,\lambda/2)$ such that $x_I(\tau)=0$.

\subsection{Proof of Theorem~\ref{thm-independent}}
We now proceed to the proof of Theorem \ref{thm-independent}. We base the proof on Lemma \ref{thm-eqdif1}. 

We first need to bound the contribution of higher order terms in the infinite sums (\ref{eq:delt}). Fix  $\epsilon>0$.
By Condition $(C_3)$, 
\begin{eqnarray*}
\lambda_S=\sum_{s\in \cS}\sum_{t \in \cT}\sum_{d=0}^{\infty}d \mu^{(s)}_{t,d}< \infty 
\end{eqnarray*}
Then, there exists an integer $\Delta_{\epsilon}$, such that $$\sum_{s\in \cS}\sum_{t \in \cT}\sum_{d=\Delta_\epsilon}^{\infty}d \mu^{(s)}_{t,d} < \epsilon,$$ which implies that for all $0 \leq \tau \leq \lambda/2$,
\begin{equation}
 \sum\limits_{s\in \cS} \sum\limits_{t\in \cT} \sum\limits_{d=\Delta_\epsilon}^{\infty}\sum\limits_{\theta=1}^d  d  \mu_{t,d}^{(s)} p_{t,d}^{(s)}(\theta)  \PP\left(\Bin(d-1,1-x) \leq \theta-1 \right)  < \epsilon.
\label{limitdelta_}
\end{equation}

Recall that the number of susceptible vertices with type $t\in \cT$, social distancing $s\in \cS$ and degree $d$ is $n_{S,t,d}^{(s)}$. Again by condition $(C_3)$, 
$$\sum_{s\in \cS}\sum_{t \in \cT}\sum_{d=0}^{\infty}d n^{(s)}_{S,t,d}/n\to \lambda_S <\infty.$$
Therefore, for $n$ large enough, $\sum_{s\in \cS}\sum_{t \in \cT}\sum_{d=\Delta_\epsilon}^{\infty}d n^{(s)}_{S,t,d}/n <\epsilon.$
and for all $0 \leq  T \leq \frac{X(0)}{2}$,
\begin{equation}
 \sum_{s\in \cS}\sum_{t \in \cT}\sum_{d=\Delta_\epsilon}^{\infty}\sum_{\theta=1}^{\infty}\sum_{\ell=0}^{\theta-1}d S^{(s)}_{t,d, \theta, \ell}(T)/n <\epsilon.
\label{limitDn}
\end{equation}

For $\Delta \geq 1$, we denote
\begin{align*}
{\bf y}^\Delta &:= \left(x_R(\tau), s^{(s)}_{t,d, \theta, \ell}(\tau)\right)_{d< \Delta, \ s\in\cS, \ 0 \leq \ell < \theta \leq d} \mbox{ and } \\
Y_n^\Delta &:= \left(X_R(T), S^{(s)}_{t,d}(T)\right)_{d< \Delta, \ s\in\cS, \ 0 \leq \ell < \theta \leq d},
\end{align*}
 both of dimension   $b(\Delta)$, and $x_R(\tau), s^{(s)}_{t,d, \theta, \ell}(\tau)$ are solutions to a system (${\rm DE}$) of ordinary differential equations.
Let $$x_*^{(\bs)} = \max\{x \in [0,1]: f^{(\bs)}(x) = x\}.$$
For the arbitrary constant $\epsilon > 0$ fixed above, we define the domain $\cD_{\epsilon}$ as
\begin{align}
D_{\epsilon}=\{\left(\tau, {\bf y}^{K_{\epsilon}} \right) \in \RR^{b(K_{\epsilon})+1} \ : \    -\epsilon <  \tau < \lambda/2 - \epsilon \ , \ -\epsilon < x_R(\tau)<\lambda, -\epsilon<s^{(s)}_{t,d, \theta, \ell}(\tau) < 1\}.
\label{uepsilon}
\end{align}
The domain $\cD_{\epsilon}$ is a bounded open set which contains the support of all initial values of the variables. Each variable is bounded by a constant times $n$ ($C_0 = 1$). By the definition of our process, the Boundedness condition is satisfied with $\beta = 1$.
The second condition of the theorem is satisfied by some $\delta_n=O(1/n)$. Finally the Lipschitz property  is also satisfied since $\lambda-2\tau$ is bounded away from zero. Then by Lemma~\ref{thm-eqdif1} and by convergence of initial conditions, we have :

\begin{corollary} \label{lem-DE}
For a sufficiently large constant $C$, we have
\begin{eqnarray}
\label{eq:diffmethod}
\PP(\forall t \leq n \sigma_H(n), \mathbf{Y}_n^{K_{\epsilon}}(t)=n\mathbf{y}^{K_{\epsilon}}(t/n)+O(n^{3/4})) = 1 - O(b(K_{\epsilon})n^{-1/4}\exp(-n^{-1/4}))
\end{eqnarray}
 uniformly for all $t \leq n\sigma_H(n)$ where
$$\sigma_H(n)=\sup\{\tau\geq 0,  d( \mathbf{y}^{K_{\epsilon}}(\tau),\partial D_{\epsilon}\ )\geq  Cn^{-1/4} \}.$$
\end{corollary}

When the solution reaches the boundary of $\cD_{\epsilon}$, it violates the first constraint, determined by $\hat{\tau} = \lambda/2 -\epsilon$.
By convergence of $\frac{X(0)}{n}$ to $\lambda$, there is a value $n_0$ such that $\forall n \geq n_0$, $\frac{X(0)}{n} > \lambda - \epsilon$, which ensures that $\hat{\tau}n\leq X(0)/2$.

Using (\ref{limitdelta_}) and (\ref{limitDn}), we have, for $0 \leq T=n\tau \leq n \hat{\tau}$ and $n \geq n_0$:
\begin{eqnarray}
\left| X_I(T)/n - x_I(\tau) \right| &\leq& |X(0)/n-\lambda| + |X_R(T)/n-x_R(\tau)| \\
&&+ \sum_{s\in \cS}\sum_{t \in \cT}\sum_{d=0}^{\infty}\sum_{\theta=1}^{\infty}\sum_{\ell=0}^{\theta-1}d |S^{(s)}_{t,d, \theta, \ell}(T)/n - s^{(s)}_{t,d, \theta, \ell}(\tau)| \nonumber\\
&\leq& \sum_{s\in \cS}\sum_{t \in \cT}\sum_{d=0}^{\Delta_\epsilon}\sum_{\theta=1}^{\infty}\sum_{\ell=0}^{\theta-1}d |S^{(s)}_{t,d, \theta, \ell}(T)/n - s^{(s)}_{t,d, \theta, \ell}(\tau)| + 3 \epsilon.
\label{limitDiff}
\end{eqnarray}
and similarly, the total number of susceptible individuals at time $T$ satisfies
\begin{eqnarray}
\left| S(T)/n - s(\tau) \right| &\leq& \sum_{s\in \cS}\sum_{t \in \cT}\sum_{d=0}^{\Delta_\epsilon}\sum_{\theta=1}^{\infty}\sum_{\ell=0}^{\theta-1}|S^{(s)}_{t,d, \theta, \ell}(T)/n - s^{(s)}_{t,d, \theta, \ell}(\tau)| + 3 \epsilon.
\label{dn},
\end{eqnarray}
where, by Lemma~\ref{lem-sol}, 
\begin{eqnarray}
s(\tau)&=&\sum_{s\in \cS}\sum_{t \in \cT}\sum_{d=0}^{\infty}\sum_{\theta=1}^{\infty}\sum_{\ell=0}^{\theta-1}s^{(s)}_{t,d, \theta, \ell}(\tau) \\
&=& \alpha_S \sum\limits_{s\in \cS} \sum\limits_{t\in \cT} \sum\limits_{d=0}^{\infty}\sum\limits_{\theta=1}^d   \mu_{t,d}^{(s)} p_{t,d}^{(s)}(\theta)  \PP\left(\Bin(d,1-x) \leq \theta-1 \right).
\end{eqnarray}

We obtain by  Corollary~\ref{lem-DE} that 
\begin{eqnarray}
\sup_{T \leq \hat{\tau} n} \left| X_I(T)/n - x_I(\tau) \right| \leq 2\epsilon + o_L(1), \ {\rm and}\\
\sup_{T \leq \hat{\tau} n} \left| S(T)/n - s(\tau) \right| \leq 2\epsilon + o_L(1).
\label{convergence}
\end{eqnarray}
We now study the stopping
time $T_n$ and the size of the epidemic $|\cR^{\bs)}(\infty)/\cR(0)|$.

Consider $x_*=\sqrt{1-2\tau_*/\lambda}$ is a stable fixed point of $f^{(\bs)}(x)$. Then by definition of $x_*$ and by using the fact that $f^{(\bs)}(1) \leq 1$, we have $f^{(\bs)}(x)>x$
for some interval $(x_*-\tilde{x},x_*)$. Then 
$$
x_I(\tau)=(\lambda x) \left(x-f^{(\bs)}(x)\right)$$
is negative in an interval $(\tau_*, \tau_* + \tilde{\tau})$. 
Let  $\epsilon$   such that $2 \epsilon < -\inf_{\tau \in (\tau_*, \tau_* + \tilde{\tau})}x_I(\tau)$ and denote $\hat{\sigma}$ the first iteration at which it reaches the minimum.
Since $x_I(\hat{\sigma}) < -2 \epsilon$ it follows that with high probability $X_I(\hat{\sigma} n)/n < 0$,
so $T_n/n = \tau_*  + O(\epsilon) + o_L(1)$. The conclusion follows by taking the limit $\epsilon\to 0$.

\section{Conclusion}\label{sec:conclusion}
We have studied a heterogeneous SIR epidemic process when a network underlies social contact.
For given social distancing strategies, we have established results on the amplification of the epidemic.  
Quantities such as the epidemic reproduction number $\fR_0$ are established and can be used as a warning signal to identify for example parts of the networks that are highly vulnerable. Vaccination and targeted social distancing can be applied  to make $\fR_0$ smaller than one.
Next, we have studied the equilibrium of the social distancing game.  Our theoretical results establish that the voluntary social distancing will always fall short of the social optimum. The social optimum itself is of course dependent on the type.

In a companion paper \cite{aminca20}, we calibrate the model to the characteristics of the Covid-19 epidemic. The dependence of the death rates on the fraction of the population that is infected plays a significant role in the gap between laissez-faire equilibrium and social optimum.

\medskip

%\paragraph{The cure.} Let $p_c$ denotes thare e probability that a cure is found after one year!  

Several directions emerge from our present study.

\medskip

%\paragraph{Optimal vaccination strategies.}
When vaccines are available but in limited supply and with challenges in distribution, an important question regards the vaccination strategy.
The social planner needs to maximize social utility under  constraints on vaccine quantity.
For a given social distancing strategy of the population, the social planner's optimization problem writes as
$$\max \sum_{i=1}^n  u_i ({\bf \ell}, {\bf s}) \ \ \text{such that} \ \ \sum_{s\in \cS}\sum_{t \in \cT}\sum_{d=0}^{\infty} \omega_{t,d} \mu_{t,d}^{(s)} \leq V $$
The main challenge remains solving a joint equilibrium problem, where agents adapt their social distancing strategy to the vaccination policy and vice-versa. In this case, a moral hazard problem may then emerge.

\medskip

%\paragraph{Individuals' optimal connectivity choice in equilibrium.}
In this paper we have considered that the degrees are given and agents choose the social distancing strategies. As in \cite{amini2019dynamic} (for a different model), we can consider the connectivity choice of every individual as resulting from a network equilibrium, in the presence of contagion risk. Each individual would then choose their connectivity (given all individuals' connectivity)
as follows
$$d^*_i=\arg\min_{d=0,1, \dots, \Delta_{t_i}} u_i ({\bf \ell}, {\bf d}),$$
where the individual utility may be written as
\begin{equation}
 u_i ({\bf \ell}, {\bf d}) =  u_i(\ell_1, ..., \ell_n, d_1, ..., d_n)  :=  \pi_{t_i}d_i - \ell_i \kappa_{t_i}\PP_n (i \in \iI^{(\bf d)} (\infty)).
\end{equation}

\medskip

Finally, we have worked in a one period social distancing game, where individuals keep their strategy constant. In a dynamic version of the game, agents (individuals) may change their strategies over time. 
 We leave this for a future work.

\bibliographystyle{acm}
\bibliography{biblio}

\end{document}